\documentclass[letterpaper,12pt]{article}%latex2e
\usepackage{graphicx}
\usepackage{amssymb}
\usepackage{amsmath}
\usepackage{epsfig}

\newtheorem{theorem}{Theorem}[section]
\newtheorem{lemma}[theorem]{Lemma}

\newenvironment{proof}[1][Proof]{\begin{trivlist}
\item[\hskip \labelsep {\bfseries #1}]}{\end{trivlist}}

\newcommand{\qed}{\nobreak \ifvmode \relax \else
      \ifdim\lastskip<1.5em \hskip-\lastskip
      \hskip1.5em plus0em minus0.5em \fi \nobreak
      \vrule height0.75em width0.5em depth0.25em\fi}

\date{}

\clearpage
\begin{document}
%\include{rstart}

%\maketitle{Nonlinear spread of rumor and inoculation strategies in the nodes with degree dependent tie strength in complex networks}
%\author{Anurag Singh, Yatindra Nath Singh }
%\date{September 1994}
\title{\bf Nonlinear spread of rumor and inoculation strategies in the nodes with degree dependent tie strength in complex networks}

%for single author (just remove % characters)
\author{
{\rm Anurag Singh, \rm Yatindra Nath Singh}\\
Department of Electrical Engineering,\\
IIT Kanpur,India-208016\\
(anuanu52@gmail.com)}
\maketitle

%\catchline{}{}{}{}{}

%%%%%%%%%%%%%%%%%%%%%%%%%%%%%%%%%%%%%%%%%%%%%%%%%%%%%%%%%%%%%%%%%%%

%\footnote{For the title, try not to
%use more than 3 lines. Typeset the title in 10 pt
%Times Roman, uppercase and boldface.}

%\author{\footnotesize Anurag Singh\footnote{Typeset names in
%10~pt Times Roman, uppercase. Use the footnote to indicate
%the present or permanent address of the author.}}
%\documentclass{article}
%\title{Cartesian closed categories and the price of eggs}
%\author{Jane Doe}
%\date{September 1994}
%\author{Anurag Singh, Yatindra Nath Singh }
%Department of Electrical Engineering,\\
% IIT Kanpur,India-208016\\
%\footnote{State completely without abbreviations, the
%affiliation and mailing address, including country. Typeset in 8~pt
%Times italic.}\\}
%anuragsg@iitk.ac.in
%

%\maketitle

%\begin{history}
%\received{(received date)}
%\revised{(revised date)}
%\accepted{(Day Month Year)}
%\comby{(xxxxxxxxxx)}
%\end{history}

\begin{abstract}
In the classic rumor spreading model, also known as susceptible-infected-removed (SIR) model in which rumor spread out of each node proportionally  to its degree and all edges  have uniform tie strength between nodes. The rumor spread phenomenon is similar as epidemic spread, in which all the informed nodes spread rumor by informing their neighbor nodes. In earlier rumor spreading models, at each time step nodes contact all of their neighbors. In more realistic scenario it is possible that a node may contact only some of its neighbors to spread the rumor. Therefore it is must in real world complex networks, the classic rumor spreading model need to be modified to consider the dependence of rumor spread rate on the degree of the spreader and the informed nodes. We have given a modified rumor spreading model to accommodate these facts. This new model, has been studied for rumor spreading in complex networks in this work. Rumor spreads nonlinearly by pairwise contacts between nodes in complex networks at degree 
dependent spreading rate. Nonlinear rumor spread exponent $\alpha$ and degree dependent tie strength of nodes affect the rumor threshold. The degree exponent tie strength between two nodes is $(k_ik_j)^\beta$, where $k_i$ and $k_j$ are degrees of node \textit{ i} and \textit{j} and $\beta$ is tie strength exponent. By using the above two exponents in any complex network gives rumor threshold as some finite value. In the present work, the modified rumor spreading model has been studied in scale free networks. It has been found to have a greater threshold than with previous classic rumor spreading model. It is also found that if $ \alpha $ and $ \beta $ parameters are tuned to appropriate value, the rumor threshold becomes independent of network size. In any social network, rumors can spread may have undesirable effect. One of the possible solutions to control rumor spread, is to inoculate a certain fraction of nodes against rumors. The inoculation can be done randomly or in a targeted fashion. We have used 
modified rumor spreading model over scale free networks to investigate the efficacy of inoculation. Random and targeted inoculation schemes have been applied. It has been observed that rumor threshold in random inoculation scheme is greater than the rumor threshold in the model without any inoculation scheme. But random inoculation is not that much effective. The rumor threshold in targeted inoculation is very high than the rumor threshold in the random inoculation in suppressing the rumor. It has also been observed that targeted inoculation scheme is more successful to stop the rumor in this model based on scale free networks. The proposed hypothesis is also verified by simulation results.
\end{abstract}

\textbf{Keywords} : Complex networks; Scale free networks; rumor spreading; nonlinear rumor spread; tie strength in complex networks; random and targeted inoculations

\section{Introduction}

Many researchers have discussed, how the properties of networks affect the dynamical process taking place in networks. In the recent years, complex network structures and their dynamics have been studied intensively \cite{Bara,BA,Nekovee1,Nekovee,Newep,New,Past,WS}. By analyzing different real world networks e.g. Internet, the www, social network and so on, researchers have identified different  topological characteristics of complex network such as the small world phenomenon and scale free property. An interesting dynamical process in complex networks is the epidemic spreading. Work on epidemic spreading has been done by many researchers \cite{Morsat,Newep}. There are two models to define the epidemic spreading. In first model- two state susceptible-infected-susceptible (SIS) model, a susceptible node can become infected and an infected node can recover and return to susceptible state, e.g. computer viruses \cite{Past2,Imm1}. The second model, is susceptible-infected-recovered (SIR) model, it defines that 
the infected nodes will become dead \cite{Morsat,Newep}. It is different from SIS model in the fact that the infected nodes will not return to the susceptible state but can attain recovered status. The epidemic spreading on complex networks has been discussed with the other models also \cite{Vespbook}. Rumor spreading is also similar in nature as epidemic spread.

 Rumors are information that circulate without officially publicized confirmation i.e. without has official reputation. Rumors are being used as a special weapon of public opinion and can pose a tremendous impact on social life. Certain rumors can affect the social stability seriously. In order to improve the resistance of the community against undesirable rumors, it is essential to develop deep understanding of the mechanism, underlying laws involved in rumor spreading, establishment of an appropriate prevention and control system to generate social stability. Sudbury studied, first time on the spread of rumors based on SIR model \cite{Sud}.  Another standard model of rumor spreading, was introduced many years ago by Daley and Kendal \cite{DK}. Its variant was introduced by Maki-Thomsan \cite{MK}. In Daley-Kendal (DK) model homogeneous population is subdivided into three groups: ignorant (who don't know rumor), spreaders (who know rumor) and stifler (know rumor but do not want to spread it). The rumor is 
propagated through the population by pairwise contacts between spreaders and other individuals in the population. Any spreader involved in a pairwise meeting attempts to infect other individual with the rumor. In case this other individual is an ignorant, it becomes a spreader. If other individual is a spreader or stifler, it finds that rumor is known and decide not to spread rumor anymore, thereby turning into stifler. In Maki Thomsan (MK) model when spreader contacts another spreader, only the initiating spreader becomes a stifler. DK and MK models have an important shortcoming that they do not take into account the topology of the underlying social interconnection networks along which rumors spread. These models are restricted in explaining real world scenario for rumor spreading. By considering the topology of network, rumor model on small world network \cite{Nekovee,Zanet,Zanet1} and scale free networks \cite{Liu} have been defined. Therefore, as long as one knows the structure of spreading networks, he 
 can figure out variables and observable to conduct quantitative analysis, forecasting and control of rumor spreading. The most important  conclusion of classical propagation theory is existence of critical point of rumor transmission intensity. When an actual intensity is greater than critical value, the rumors can spread in networks and persistently exist. When the actual intensity is less than the critical value, rumors decay at an exponential rate and this critical value is called rumor threshold. In previous studies of rumor spreading, underlying network topologies were scale free  which is also observed in many real networks. But some more properties of real world networks have not been considered. In the earlier rumor spreading models, rumor transmission rate was fixed. In the real scenario, it will be different among different nodes. It will depend on tie strength between nodes and node's strength. In Internet, weight implies the knowledge of its traffic flow or the bandwidth of routers \cite{Int}, 
in the world wide airport networks it can define the importance of an airport \cite{Barrat} and so on. In case of rumor spreading, the strength can indicate the frequency of the contact between two nodes in scale free networks. Greater the strength, the more intensely the two nodes are communicating. Chances of spreading rumors tend to differ among individuals  and laws of spreading in social networks with different topologies are also different. Each informed node can make contacts with all of its neighbors in a single time step. In other words we can say that each informed to number of nodes, can spread information to nodes equal to its degree. In real case, an informed node can't make contact to all of its neighbors in single time step. Studies on small world networks found that compared with regular network, small world network has smaller transmission threshold and faster dissemination. Even at small spreading rates, rumors can exist for long. Studies on infinite-size scale free networks have also 
revealed that no matter how small transmission intensity be, the rumors can be persistent as positive critical threshold does not exist.\cite{Pastep,Past}.
 	
 	In previous studies on rumor spreading in scale free networks, it has been assumed that larger the nodal degree, the greater the rumor spreading from the informed node, i.e. the rumor spread is proportional to the nodal degree. With these assumptions for SIR model, for scale free networks with sufficiently large size, the rumor threshold  $\lambda_c=0$ can be zero. Yan et al.  have demonstrated that the asymmetry of infection plays an important role \cite{Yan}. They redistribute the asymmetry to balance the degree heterogeneity of the network and found finite value of  epidemic threshold. Zhou et al. concluded that this hypothesis is not always correct  \cite{Zhou,Zl}. In rumor spreading contact networks, the hub nodes have many acquaintances; however they cannot contact all their acquaintances in single time step. They assumed that the rumor spreadness is not equal to the degree but identical for all nodes of the scale free networks and obtained the threshold $\lambda_c=\frac{1}{A}$, where \emph{A} is the 
constant infectivity of each node and is not equal to the degree of node. Recently, Fu et al. \cite{Fu} have defined piecewise linear infectivity. They suggested if the degree \textit{k}, of a node is small, its infectivity is $ \alpha k $ otherwise its infectivity is a saturated value \textit{A} when \textit{k} is beyond a constant $ A/\alpha $. In both constant and piecewise linear infectivity, the heterogeneous infectivity of the nodes due to different degrees has not been considered. While in scale free networks heterogenity in nodal degree is very common. There may be nodes with different degrees, which have the same infectivity, and there will be a large number of such nodes if  infectivity does not saturate or the size of network is infinite.

% In scale free networks, a rumor with an arbitrary rate  of spread will exist for a long time. In these networks, rumors first affect individuals who have more social contacts, then the general individuals and finally those with less social contact. It has been found that in scale free networks, rumors spread at a relatively low speed for a very short period of time starting from the outbreak and then rise rapidly to a high peak, followed by a rapid decline in exponentially. Studies of rumor spreading on complex networks are in ascendant and the results have largely changed the views on the issue of rumor spreading, which play an important role in the prevention and control against rumors in such events. The theory and the method of transmission dynamics being applied to the analysis of structure and characteristics of rumor spreading play a vital role in the design of rumors prevention and control system.

% In order to make transmission rate as per the real case, weights of edges and strength of nodes play an important role in the scale free networks. The strength is one of the important point in many real networks. For rumor spreading, the strength can define the extent of frequency of contracting of two nodes in scale free networks; greater the strength, the more intensively two nodes communicate \cite{Barrat}. 

 In scale-free networks, a small number of nodes have very high degrees. Although, random immunization strategy works very well in homogeneous random networks, but this strategy is not effective in preventing a rumor in scale free networks. Hence a new immunization strategy needs to be developed which is able to recover the rumor threshold. One of the most efficient approach is to immune the highest degrees nodes, or, more specifically, to immune those nodes (hereafter termed as hubs or hub nodes) which have degrees higher than a preset cut-off value $k_c$ . Such a strategy is known as targeted immunization \cite{Imm2,CohenRes,Cohenbr,Madhav,Pastep,Imm1,Anu}.
  
  In order to control the epidemic spread of rumors, inoculating the nodes is an option. Random inoculation was found to be ineffective for all the complex networks (includes scale free networks) by Pastor-Satorras and Vespignani \cite{Past}. On the other hand, targeted inoculation is successful in arresting the rumor spread in scale free networks \cite{Past}. Acquaintance inoculation method proposed by Cohen et.al. \cite{Madar} was found effective for both scale free networks as well as networks with bimodal degree distribution. Site percolation is another process inversely similar to inoculation [19]. Inoculated node can be considered to equivalent of removed site in site percolation. After the percolation threshold, the network will have maximum size of giant component \cite{Newep}. Thus if the network operates below the percolation threshold, the rumor spread will not happen as outbreak. In other worlds, inoculation ensures the operation of network below percolation threshold. Random inoculation usually 
requires inoculation of large number of nodes for being effective. If nodes with higher connectivity are targeted for inoculation, the same effectiveness can be achieved with smaller number of inoculated nodes. But it needs knowledge of nodes which have higher connectivity \cite{Madar}.
  
In this work, we have investigated rumor spread for the scale free network while considering varying tie strengths between nodes. Further we have assumed that a non linearly varying number of neighbors are infected with the rumor in each time step. While in the earlier models \cite{Morsat,Nekovee}, tie strength has been considered to be uniform, and a constant number of neighbors have been assumed to be infected in each time step by each node. If a node has \textit{K} neighbors, in the earlier models, in each time step, all the \textit{K} neighbors will be infected. We have modified the earlier SIR model given by Nevokee \cite{Nekovee} and included a rumor spreading exponent $ \alpha $ . In this work, $ K^\alpha $  neighboring nodes will be infected in each time step. Here $ \alpha $ is the spreading exponent where, $ 0<\alpha \leq 1 $. We have used Barabasi-Albert (BA) model \cite{BA} to create scale free networks with power law distribution of nodal degree, and then used the proposed strategy in them to 
study the rumor spread. Scale free networks have been specifically chosen as they are much more heterogeneous than the small world or the random network models, thus forming a good candidate for testing our proposition.

 The dynamical differential equations have been used to represent the modified model for information spread. The equations have been used to find the threshold and study the rumor propagation behavior. The result have also been verified by the simulations
  %     The results of the presented modified model will give some insight into the spread of real rumor spreading over scale free networks. A modified SIR model is given. The classical SIR model, as proposed by Nekovee on complex networks for rumor spreading is modified with nonlinear rumor spread exponent $ \alpha $ and the tie strength exponent $ \beta $ \cite{Nekovee}. In this modified model, the dynamical differential equation for information spreading has been proposed.  
By choosing the appropriate values of $ \alpha $ and $ \beta $, finite non zero threshold value can be found for the scale free networks. It is found that the $\alpha  $ is more sensitive than the $ \beta $ for spreading rumor in a large scale free network. Finally, the rumor threshold has been calculated after applying the random and targeted inoculation of nodes to suppress the harmful rumor in the scale free networks. In the targeted inoculation scheme, the rumor threshold has been found to be larger. The targeted inoculation has also been effective in the scale free networks to suppress rumors.

\section{ Classical SIR model} 
Classical SIR model is one of the most investigated rumor spreading models for complex networks. In this model, nodes are in one of the three categories \textemdash ignorants (the nodes who are ignorant of the rumor), spreaders (those who hear the rumor and also actively spread it) and stiflers (the nodes who hear the rumor but do not spread it further). The rumor is propagated through the nodes by pairwise contacts between the spreaders and other nodes in the network. Following the law of mass action, the spreading process evolves with direct contact of the spreaders with others in the population. These contacts can only take place along the edges of undirected graph of complex network. If the other node is the spreader or stifler then the initiating spreader becomes the stifler. 
%$G=(V,E)$, where \emph{V} and \emph{E} mention the nodes and the edges of the network, respectively.
The classical SIR model has been studied by M. Nevokee et
al.\cite{Nekovee1,Nekovee} for heterogeneous population (nodes have different
degrees). In this paper $I(k,t), S(k,t), R(k,t)$ are the expected values of
ignorants, spreaders and stifler nodes in network with degree $k$ at time $t$.
Above rumor spreading process can be summarized by following set of pairwise
interactions.\\
\begin{eqnarray*}
&&I_1+S_2\xrightarrow{\lambda}S_1 + S_2,\\
&&\mbox{(when spreader meets with the ignorant, it becomes spreader at rate
$\lambda$)}\\
&&S_1+R_2\xrightarrow{\delta}R_1+R_2,\\
&&\mbox{(when a spreader contacts with stifler, the spreader}\\
&&\mbox{becomes a stifler at the rate $\delta$)}\\
&&S_1+S_2\xrightarrow{\delta} R_1+S_2,\\
&&\mbox{(when a spreader contacts with another spreader, initiating spreader}\\
&&\mbox{becomes a stifler at the rate $\delta$)}\\
&&S\xrightarrow{\sigma} R.\\
&&\mbox{($\sigma$ is the rate to stop spreading of a rumor spontaneously.)}
\end{eqnarray*}
Let $\rho^i(k,t)=I(k,t)/N(k), \rho^s(k,t)=S(k,t)/N(k), \rho^r(k,t)=R(k,t)/N(k)$
are the fraction with respect to total nodes, of ignorant, spreaders and stifler
nodes respectively with degree \textit{k} at time \textit{t}. These fractions of
the nodes satisfy the normalization condition $
\rho^i(k,t)+\rho^s(k,t)+\rho^r(k,t)=1$ where, $N(k)$ represents the total number
of nodes with degree $k$, in the network. Nevokee et al. \cite{Nekovee} proposed
the formulation of this model for analyzing complex networks as interacting Markov
chains. They used the framework to derive from the first-principles, mean-field
equations for the dynamics of rumor spreading in complex networks with arbitrary
correlations. These are given below.
 \begin{eqnarray}
&&\frac{d\rho^i(k,t)}{dt}=-k\lambda
\rho^i(k,t)\sum_{l}P(l|k)\rho^s(l,t);\label{1-1}\\
&&\frac{d\rho^s(k,t)}{dt}=k\lambda\rho^i(k,t)\sum_{l}P(l|k)\rho^s(l,
t)-k\delta\rho^s(k,t) \sum_{l}\left ( \rho^s(l,t) \right . \nonumber\\
&&\left . +\rho^r(l,t)\right ) P(l|k)-\sigma\rho^s(k,t) ;\label{1-2}\\
&&\frac{d\rho^r(k,t)}{dt}=k\delta\rho^s(k,t)\sum_{l}(\rho^s(l,t)+\rho^r(l,
t))P(l|k)+\sigma\rho^s(k,t);\label{1-3}
\end{eqnarray}

Where conditional probability $P(l/k)$ is the degree-degree correlation function
that a randomly chosen edge emanating from a node of degree \emph{k} leads to a
node of degree \emph{l}. Here, it has been assumed that the degree of nodes in
the whole network are uncorrelated. Therefore, degree-degree correlation is $
P(l|k)=\frac{lP(l)}{ \langle k\rangle}$ where $P(l)$ is the degree distribution
and  $\langle k\rangle$ is the average degree of the network. Nevokee \emph{et
al.} have shown that the critical threshold for rumor spreading is independent
of the stifling mechanism. The critical threshold found by him was
$\lambda_c=\frac{\langle k\rangle}{\langle k^2\rangle}$. It is same as found
for SIR model \cite{Lloyd,Morsat}. Hence, it implies  epidemic threshold is
absent in large size scale free networks $(\langle k^2\rangle\rightarrow \infty,
\lambda_c\rightarrow 0)$. This result is not good for epidemic control, 
since the epidemics will exist in the real networks for any non zero value of
spreading rate $\lambda$.

%\section{Scale free networks}
The ``scale-free network" (Internet, www etc.) \cite{Bara,BA} describes the degree distribution as a power law, i.e., the probability of a node having  degree \textit{k} 
\begin{equation*}
P(k)\propto k^{-\gamma}
\end{equation*}
Where, $ 2<\gamma \leq 3 $. Degree distribution is plotted on a log - log scale where, $\gamma$ is the slope of the degree distribution, which is a straight line. 
 One can construct a scale-free network by starting with a few nodes $ m_0 $and  then by  adding a new node with \textit{m} ($ m\leq m_0 $) edges  per node at each time step. The probability $ \Pi_v $ of node $ v $ being connected to a new node depends on its degree $ j_v $. The higher its degree, the greater the probability ($ \Pi_v=j_v/\sum_{k} j_k $) that it  will get new connections (preferential attachment).

\section{Modified rumor spreading model}
The real world networks can have the intimacy, confidence etc. between the nodes. Unlike previous studies where each node can spread the rumor with constant transmission rate $\lambda$. In this study, we have concentrated the rumor spreading model with nonlinear rumor spread and transmission rate between two connected nodes is a function of their degrees. Here, to remove the complexity stifling rate $ \delta $ is being considered 0. Based on this assumption, we define,
\begin{eqnarray}
&&\frac{d\rho^i(k,t)}{dt}=-k \rho^i(k,t)\sum_{l}P(l|k)\rho^s(l,t)\frac{\Theta(l)}{l} \lambda_{lk};\\\label{5-1}
&&\frac{d\rho^s(k,t)}{dt}=-\sigma \rho^s(k,t)+ k(1-\rho^s(k,t)-\rho^r(k,t))\sum_{l}P(l|k)\rho^s(l,t)\frac{\Theta(l)}{l} \lambda_{lk} ;\\\label{5-2}
&&\frac{d\rho^r(k,t)}{dt}=\sigma \rho^s(k,t).\label{5-3}
\end{eqnarray}
where $\Theta(k)$ represents the rumor spreadness of a node with degree \textit{l} and $\lambda_{lk}$ represents the rumor spreading rate from nodes of degree \textit{l} to nodes with degree \textit{k}.

\subsection{Tie strength in complex networks}
 The topological properties of a graph are fully encoded in its adjacency matrix \textbf{A} , whose elements $a_{ij}$ ($ i\neq j $) are 1 if a link connects node \textit{i} to node \textit{j}, and 0 otherwise. The indices \emph{i}, \emph{j} run from 1 to \emph{N}, where \emph{N} is the size of the network. Similarly, a weighted network is entirely described by a matrix \textbf{W} whose entry $w_{ij}$ gives the weight on the edge connecting the vertices \emph{i} and \emph{j} ($w_{ij}$ = 0, if the nodes \textit{i} and \textit{j} are not connected). In this study we will consider only the case of symmetric weights ($w_{ij} = w_{ji}$) while the undirected case of the network is considered \cite{Barrat}. In a call network  if two nodes call each other for a long duration then weight of the connecting edge will be high and it shows high tie strength between them \cite{Onnela}. Here weight of the edge in terms of total call duration defines the tie strength between the nodes. It has also been  observed in the 
dependence of the edge weight $w_{ij}$ to define strength between nodes with end point degrees $k_i$ and $k_j$. Weight as a function of the end-point degrees can be well approximated by a power-law dependence,
 \begin{equation*}
 w_{ij} =b(k_ik_j)^\beta\\
 \end{equation*}
 where, $ \beta $ is the  degree influenced real exponent which depends on the type of complex networks and \textit{b} is a positive quantity. When $ \beta>0 $ then rumor transmit to high degree nodes and when $\beta<0  $ then rumor will transmit to low degree nodes. Further, if $ \beta=0 $ there will be degree independent transmission.

% Spreading rate $\lambda$ shows the property of nodes to accept rumor with certain probability, depends on the credibility of rumor. Stifling rate $\alpha$ shows the property of individuals to loose interest in spreading of rumor, when they got it through contact with others. In classical SIR model for heterogeneous population can be described as the following set of mean field equations:
It has been observed that the individual edge weight doesn't provide clear view of network's complexity. A detailed measurement if tie strength using the actual weights is obtained by enhancing the property of a vertex degree $ k_i=\sum_j a_{ij} $ in terms of the vertex strength $ S_i=\sum_{j=1}^{N} a_{ij}w_{ij} $ (total weights of their neighbors). Therefore, there is a coupling between interaction strengths of the nodes with the counterintuitive consequence that social networks are robust enough to the removal of the strong ties but fall apart after a phase transition if the weak ties are removed \cite{Barrat}. Therefore, we can measure the strength of a node of degree \textit{k} for scale free network,
\begin{eqnarray}
S_k &=& k\sum_l P(l|k)w_{kl}\nonumber\\
&=& k\sum_l \frac{lP(l)}{\langle k \rangle}w_{kl}\nonumber\\
&=& b\frac{k^{1+\beta}}{\langle k \rangle}\langle k^{1+\beta}\rangle \label{5-4}
\end{eqnarray}

Here, it has been considered that the rumor spreading model, where rumor transmission rate in contact process between a spreader node and an ignorant node influenced by their degrees. If $ w_{kl} $ is the tie strength between \textit{k}-degree node and \textit{l}-degree node for (\textit{k},\textit{l}) edge, $ S_k $ is the node strength with degree \textit{k}. In scale free network each node of degree \textit{k}, there is a constant rumor transmission rate $ \lambda k $. Therefore, rumor transmission rate from \textit{k}-degree node to \textit{l}-degree node is given by the proportion of $ w_{kl} $ to $ S_k $. Hence, $ \lambda_{kl} $ can be defined as,
\begin{equation}
\lambda_{kl}=\lambda k \frac{w_{kl}}{S_k} \label{5-5}
\end{equation}
We can see in Eq. \eqref{5-5} that by increasing the proportion of $ w_{kl}/S_k $, the more possibility of rumor transmission rate can be increases through the edge. In the present work, uncorrelated networks have been considered, hence  $\lambda_{kl}=\lambda l^\beta\langle k\rangle/\langle k^{1+\beta}\rangle$.
%The rumor spreading in scale free networks is not symmetric always, i.e. the traffic from node \emph{i} to one of its neighbor node \emph{j} ($w_{ij}$)does not always equal that to its another neighbor \emph{l} ($w_{il}$). In most of the previous researches it has been assumed that $w_{ij}=w_{il}$ for any two neighbors \emph{j} and \emph{l} of node \emph{i}. Here asymmetrical rumor propagation for rumor spreading threshold and phase transition has been considered.
%$w_{kl}=w_0(kl)^\beta$ $S_k=k\sum_{l}P(l/k)w_{kk'}$   $P(l/k)=lP(l)/\langle k \rangle$   $S_k=w_0\frac{\langle k^{1+\beta}\rangle k^{1+\beta}}{\langle k\rangle}$
%\begin{equation}
%\lambda_{kl}=\lambda k\frac{w_{kl}}{S_k}
%\end{equation}
%\begin{equation}
%\Theta(k)=k^\alpha, 0<\alpha \leq 1
%\end{equation}
%\begin{eqnarray}
%\Theta(k)=k^a, 0<a\leq1
%\end{eqnarray}
In this model, rumor spreadness, $ \Theta(k)=k^\alpha $ where $ 0<\alpha \leq1 $, it defines that each spreader node may contact with $ k^\alpha $ neighbors within one time step. Therefore, spreadness of a rumor will vary nonlinearly with the growing degree \emph{k}.
 In Eqs. \eqref{5-1}-\eqref{5-3} by putting $ \sigma=1 $ (without the lose of generality), we can write rumor equation for $ \Theta(k) $ and $ \lambda_{lk} $ can be written as,
 \begin{eqnarray}
&&\frac{d\rho^i(k,t)}{dt}=-\frac{\lambda k^{1+\beta}}{\langle k^{1+\beta}\rangle}\rho^i(k,t)\sum_l l^\alpha P(l)\rho^s(l,t);\\\label{5-6}
&&\frac{d\rho^s(k,t)}{dt}=-\rho^s(k,t) + \frac{\lambda k^{1+\beta}}{\langle k^{1+\beta}\rangle}\rho^i(k,t)\sum_l l^\alpha P(l)\rho^s(l,t);\\\label{5-7}
&&\frac{d\rho^r(k,t)}{dt}=\rho^s(k,t)\label{5-8}
\end{eqnarray}
 Here,  $ \Phi(t)=\sum_{k}k^\alpha P(k)\rho^s(k,t) $ has been used as auxiliary function. After solving rumor Eqs. \eqref{5-6}-\eqref{5-8} with initial conditions $ \rho^i(k,0)\simeq  1, \rho^s(k,0)\simeq  0, \rho^r(k,0)\simeq  0 $ we get,
\begin{equation}
\rho^i(k,t)=e^{\frac{-\lambda k^{1+\beta}}{\langle k^{1+\beta}\rangle}\Psi(t)}\label{5-9} 
\end{equation}
where,
\begin{equation}
\Psi(t)=\int_{0}^{t}\Phi(t)dt=\sum_{k}k^\alpha P(k)\rho^r(k,t)\label{5-10}
\end{equation}
 
\section{Rumor threshold of the modified model}
Time derivative of Eq. \eqref{5-10} can be obtained by, 
\begin{eqnarray}
\frac{d\Psi(t)}{dt}&=&\sum_{k}k^\alpha P(k)\frac {d\rho^r(k,t)}{dt}\nonumber\\
\mbox{from Eq. \eqref{5-10},}\nonumber\\
&=&\sum_{k}k^\alpha P(k)\rho^s(k,t)\nonumber\\
&=&\sum_{k}k^\alpha P(k)(1-\rho^r(k,t)-\rho^i(k,t))\nonumber\\
&=&\sum_{k}k^\alpha P(k)-\sum_{k}k^\alpha P(k)\rho^r(k,t)-\sum_{k}k^\alpha P(k)\rho^i(k,t))\nonumber\\
&=&\langle k^{\alpha}\rangle -\Psi(t)-\sum_{k}k^\alpha P(k)\rho^i(k,t))\nonumber\\
&=&\langle k^\alpha \rangle-\Psi(t)-\sum_{k}k^\alpha P(k)e^{-\frac{\lambda k^{1+\beta}}{\langle k^{1+\beta}\rangle}\Psi(t)} \label{6-1}
\end{eqnarray}
 In the infinite time limit, i.e., at the end of rumor spreading, we will have $S_k(\infty)$=0, $lim_{t\rightarrow \infty}\Psi(t) \rightarrow \psi$ and $lim_{t\rightarrow \infty}d\Psi/dt$=0,
\begin{eqnarray}
0&=&\langle k^\alpha \rangle-\Psi-\sum_{k}k^\alpha P(k)e^{-\frac{\lambda k^{1+\beta}}{\langle k^{1+\beta}\rangle}\Psi}\nonumber\\
\Psi&=&\langle k^\alpha \rangle-\sum_{k}k^\alpha P(k)e^{-\frac{\lambda k^{1+\beta}}{\langle k^{1+\beta}\rangle}\Psi}\label{6-2}
\end{eqnarray}

We can conclude that $ \Psi=0 $ will always be the solution of above equation. It has been found that right hand side of Eq. \eqref{6-2} is a convex function. Therefore the derivative of Eq. \eqref{6-2} with respect to $ \Psi $ will always be in order to achieve greater than one to get non zero solution,
\begin{equation}
\frac{d}{d\Psi}(\langle k^\alpha \rangle-\sum_k k^\alpha P(k)e^{-\frac{\lambda k^{1+\beta}}{\langle k^{1+\beta}\rangle}\Psi})|_{\Psi=0} >1 \label{6-3}
\end{equation}

Hence,
\begin{equation}
\sum_k P(k)\lambda \frac{k^{\alpha+\beta+1}}{\langle k^{1+\beta}\rangle}=\lambda \frac{\langle k^{\alpha+\beta+1}\rangle}{\langle k^{1+\beta}\rangle}>1\label{6-4}
\end{equation}
Using above equation, the rumor threshold can be defined as,
\begin{equation}
\lambda_c= \frac{\langle k^{\beta+1}\rangle}{\langle k^{\alpha+\beta+1}\rangle}\label{6-5}
\end{equation}

It is interesting to note that, by putting, $ \alpha=1 $ and $ \beta=0 $ in Eq. \eqref{6-5}, the threshold for this model reduces to $ \langle k \rangle/\langle k^2 \rangle $ for classical rumor spreading model\cite{Past}.

When, $ t \rightarrow \infty $ spreader nodes will be 0, ($ \rho^s(k,\infty)=0 $) and from Eq. \eqref{5-9}, $\rho^i(k,\infty)=e^{\frac{-\lambda k^{1+\beta}}{\langle k^{1+\beta}\rangle}\Psi}$. Therefore, final size of rumor \textit{R} at $ t \rightarrow \infty$ ($ lim_{t\rightarrow \infty}\rho^r(k,t)=R $),
\begin{eqnarray}
R &=&\sum_k P(k)\rho^r(k,\infty)\\
&=&\sum_k P(k)(1-\rho^s(k,\infty)) \nonumber \\
&=&\sum_k P(k)(1-e^{\frac{-\lambda k^{1+\beta}}{\langle k^{1+\beta}\rangle}\Psi}) \nonumber \\
&=&1-\sum_k P(k)e^{\frac{-\lambda k^{1+\beta}}{\langle k^{1+\beta}\rangle}\Psi}\\ \label{6-19}
\end{eqnarray}

In the real world complex networks rumor spreads on a finite size complex networks. It may be possible that size of scale free network is very large.The maximum or minimum degree of scale free network is mentioned by $ k_{max}$ or $k_{min} $. Pastor et al. found that the epidemic threshold $ \lambda_c$ for $k_{max} $ for SIS model on bounded SF networks with $P(k)\sim k^{-2-\gamma^\prime}$, $ 0<\gamma^\prime \leq 1 $. They assumed that with the soft and hard cut-off $ k_{min}$ and $ k_{max} $, when $ \alpha=1 $. The hard cut-off denotes that, a network does not possess any node with degree $ k>k_{max}$. As $ k_{max} $ of a node is network age, defined in the terms of number of nodes \textit{N},
\begin{equation}
k_{max}=k_{min}N^\frac{1}{\gamma^\prime +1 }\label{6-6}
\end{equation}
The normalized degree distribution is defined by,
\begin{equation}
P(k)=\frac{(1+\gamma^\prime)k_{min}^{1+\gamma^\prime}}{1-(k_{max}/k_{min})^{-1-\gamma^\prime}}k^{-2-\gamma^\prime}\theta(k_{max}-k)\label{6-7}
\end{equation}
Where $ \theta(x) $ is a heavy side step function \cite{Pastep} .

In modified rumor spreading model if $\alpha=1$ and $\beta=0 $ then it will be classic rumor spreading model and degree distribution in scale free networks $ P(k)=k^{-\gamma} $ where, $ 2\leq \gamma \leq 3 $, therefore  
\begin{eqnarray}
{\lambda_c}'(k_{max})&=&\frac{\langle k \rangle}{\langle k^2 \rangle}\\\label{6-8}
&=&\frac{\int_{k_{min}}^{k_{max}} k^{1-\gamma} dk}{\int_{k_{min}}^{k_{max}}k^{2-\gamma} dk}\\\label{6-9}
&\simeq &\frac{3-\gamma}{{(\gamma-2)} k_{min}}(k_{max}/k_{min})^{\gamma-3} \label{6-10}
\end{eqnarray}
Eq.\eqref {6-6} is modified for the given scale free network as,
\begin{equation}
k_{max}=k_{min}N^\frac{1}{\gamma -1 }\label{6-new}
\end{equation}
\begin{equation}
{\lambda_c}'(N)\simeq \frac{3-\gamma}{{(\gamma-2)} k_{min}}(N)^{(\gamma-3)/(\gamma-1)}\label{6-11}
\end{equation}
for $ \gamma =3 $,
\begin{equation}
{\gamma_c}'(N)\simeq 2[k_{min}ln(N)]^{-1} \label{6-12}
\end{equation}
Eqs. \eqref{6-11}-\eqref{6-12} shows that $ \lambda^\prime_c \rightarrow 0 $ if $ N\rightarrow \infty $

In modified rumor spreading model, nonlinear rumor spread is considered with  $ \Theta(k)=k^{\alpha} $ ,
\begin{eqnarray}
{\lambda_c}^\#(k_{max}) &=&\frac{\int_{k_{min}}^{k_{max}}k^{\beta+1-\gamma} dk}{\int_{k_{min}}^{k_{max}}k^{\alpha+\beta+1-\gamma} dk} \nonumber\\
&=& k_{min}^{(-\alpha)}\frac{\alpha+\beta-\gamma+2}{\beta-\gamma+2}\frac{[(k_{max}/k_{min})^{\beta-\gamma+2}-1]}{[(k_{max}/k_{min})^{\alpha+\beta-\gamma+2}-1]}\label{6-14}
\end{eqnarray}
\begin{theorem}
In classic rumor spread model ($ \alpha=1, \beta=0 $) threshold is smaller than the modified rumor spread model ($ 0<\alpha<1$ and $\beta \neq 0 $).
\end{theorem}
The proof is given in the next section after lemmas.

\begin{lemma}
When the size of network (\textit{N}) increases, the value of critical threshold $ \lambda^\#_c > 0$ for $ \alpha+\beta+2<\gamma $, otherwise it will approaches to 0.
\end{lemma}
\begin{proof}
Since $ k_{max}/k_{min}=N^{1/{\gamma-1}} $, therefore  $ k_{max}/k_{min}$ increases when \textit{N} increases, it becomes infinity when $ N \rightarrow \infty $. When, $ \alpha+ \beta+2<\gamma, (\frac{k_{max}}{k_{min}})^{\beta -\gamma+2} =(\frac{k_{max}}{k_{min}})^{\alpha +\beta -\gamma+2} = 0$. The value of $ \lambda^\#_c $ will be positive. Here, $ \beta<0 $ (rumor transmission influenced to low degree nodes) is considered. Now from Eq. \eqref{6-14}, we can conclude that $ \lambda^\#_c $ will be positive . For $ \alpha+ \beta+2\geq \gamma$, $ \lambda^\#_c \rightarrow 0 $ when \textit{N} increases. It can be summarized as,
\begin{equation}
\hspace*{0.8cm}{\lambda_c}^\#(k_{max})\ =\ %\left \{
  \begin{cases}%{array}{cc}
  k_{min}^{(-\alpha)}\frac{\alpha+\beta-\gamma+2}{\gamma-\beta-2}(k_{max}/k_{min})^{\gamma- \alpha-\beta-2}, & \ \alpha+ \beta + 2 > \gamma\\
   k_{min}^{(-\alpha)}\frac{\gamma-\alpha-\beta-2}{\gamma-\beta-2}, & \ \alpha+ \beta +2< \gamma \\
   k_{min}^{(-\alpha)}\frac{1}{\alpha ln(k_{max}/k_{min})} , & \ \alpha+ \beta+2 = \gamma
    \end{cases}\label{6-15}
    \end{equation}
\end{proof}
   
%\begin{lemma}
%Since $ k_{max/k_{min}}=N^{1/{\gamma-1}} $ therefore $ k_{max/k_{min}}$ increases when N increases, it becomes infinity when $ N \rightarrow \infty $. When, $ \alpha+ \beta+2<\gamma, (\frac{k_{max}}{k_{min}})^{\beta -\gamma+2} =(\frac{k_{max}}{k_{min}})^{\alpha +\beta -\gamma+2} = 0$. The value of $ \lambda^\# $ will be positive. In other cases it will tends to 0.
%\end{lemma}
\begin{lemma}
In given rumor spreading model when $\alpha+ \beta +2< \gamma$ then rumor spreading threshold $ \lambda^\# $ is independent from the size of scale free network (\textit{N}).
\begin{proof}
It may also be defined using Eqs. \eqref{6-new}-\eqref{6-15}  in the term of the number of nodes \textit{N},

\begin{equation}
\hspace*{0.8cm}{\lambda_c}^\#(N)\ =\ %\left \{
  \begin{cases}%{array}{cc}
  k_{min}^{(-\alpha)}\frac{\alpha+\beta-\gamma+2}{\gamma-\beta-2}(N)^{(\gamma- \alpha-\beta-2)/(\gamma-1)}, & \ \alpha+ \beta+2> \gamma\\
   k_{min}^{(-\alpha)}\frac{\gamma-\alpha-\beta-2}{\gamma-\beta-2}, & \ \alpha+ \beta+2 < \gamma \\
   k_{min}^{(-\alpha)}\frac{\gamma-1}{\alpha ln(N)} , & \ \alpha+ \beta = \gamma
    \end{cases}\label{6-16}
    \end{equation}
  Here, it is found that for $ \alpha+\beta+2<\gamma $, $ \lambda^\#_c $ is independent of \textit{N}.
\end{proof}
\end{lemma}

    \begin{proof}
     Now using lemmas (1) and (2)the theorem can be proved for $\alpha+\beta-\gamma+2>\gamma$. The ratio of rumor threshold in classic model and given model is given as,
%      It has been observed that positive critical value ${\lambda_c}^\#(N)  $ is not related to the size of network when$  \alpha+ \beta < \gamma $ ,therefore
%      \begin{equation}
%      \lambda_c=\frac{\alpha+ \beta -\gamma+2}{\beta -\gamma+2}\frac{{k_{max}}^{\beta-\gamma+2}-{k_{min}^{\beta-\gamma+2}}}{{k_{max}}^{\alpha+\beta-\gamma+2}-{k_{min}^{\alpha+\beta-\gamma+2}}}\label{6-17}
%      \end{equation}
%      if $ k_{max} \rightarrow \infty or N \rightarrow \infty$ than $\lambda_c \rightarrow 0  $,
%          
%    \end{lemma}

\begin{equation}
\frac{{\lambda_c}(N)}{{\lambda_c}^\#(N)}=\frac{(2-\gamma)(\gamma-\beta-2)}{(\gamma-2) k_{max}^{(1-\alpha)}(\alpha+\beta-\gamma+2)N^(1-\gamma-\beta+2)/(\gamma-3)}\\\label{6-18}
\end{equation}

It has been found from Eq. \eqref{6-18} that  $\frac{{\lambda_c}(N)}{{\lambda_c}^\#(N)} < 1$ for finite scale free networks. Therefore, it has been justified that rumor threshold $ \lambda^\#_c(N) $ is greater than the $ \lambda_c(N) $ in finite size scale free networks. In finite size scale free networks, when $ 0<\alpha <1 $, $ \beta \neq 0 $ and $ \alpha+\beta+2>\gamma $ then it is  hard to spread rumor in comparison to networks have $ \alpha=1 $ and $ \beta=0 $. Finite rumor threshold is possible for any size of networks as seen in Eq. \eqref{6-16}. However, it will be 0 when \textit{N} approaches infinity. 
\end{proof}

\section{Random inoculation}
 In random inoculation strategy, randomly selected node will be inoculated. This approach inoculates a fraction of the nodes randomly, without any information of the network. Here variable \textit{g} ($0\leq g\leq 1$) defines the fraction of inoculative nodes. In the presence of random inoculation rumor spreading rate $ \lambda $ reduced by a factor $ (1-g) $. In mean field level, for the scale free networks in the case of random inoculation, the rumor equations are modified using initial conditions as,
 \begin{eqnarray}
 &&\frac{d\rho^i(k,t)}{dt}=-\frac{\lambda(1-g) k^{1+\beta}}{\langle k^{1+\beta}\rangle}\rho^i(k,t)\Phi(t);\\\label{7-1}
&&\frac{d\rho^s(k,t)}{dt}=-\rho^s(k,t) + \frac{\lambda(1-g) k^{1+\beta}}{\langle k^{1+\beta}\rangle}\rho^i(k,t)\Phi(t);\\\label{7-2}
&&\frac{d\rho^r(k,t)}{dt}=\rho^s(k,t).\label{7-3}
\end{eqnarray}

 Now $\Psi$ is modified as (from Eq. \eqref{6-2}),
 \begin{eqnarray}
 \Psi&=&\langle k^\alpha \rangle-\sum_{k}k^\alpha P(k)e^{\frac{-\lambda(1-g) k^{1+\beta}}{\langle k^{1+\beta}\rangle}\Psi} \label{7-4}
 \end{eqnarray}
%  where, $P(k)$ will be the degree distribution of scale free network after inoculating \textit{g} fraction of nodes. 
Therefore, final size of informed nodes (\textit{R}) is,
  \begin{eqnarray}
  R = 1- \sum_{k}P(k)(1-g)e^{\frac{-\lambda(1-g) k^{1+\beta}}{\langle k^{1+\beta}\rangle}\Psi}-g \label{7-5}
  \end{eqnarray}
  To achieve nontrivial solution of $ \Psi $ from Eq. \eqref{7-4},
  \begin{equation}
\frac{d}{d\Psi}(\langle k^\alpha \rangle-\sum_k k^\alpha P(k)e^{\frac{-\lambda (1-g)k^{1+\beta}}{\langle k^{1+\beta}\rangle}\Psi})|_{\Psi=0} >1 \label{7-6}
\end{equation}
Therefore, rumor spreading threshold in the case of random inoculation is obtained as,
\begin{equation}
\hat{\lambda_c}= \frac{\langle k^{\beta+1}\rangle}{(\langle k^{\alpha+\beta+1}\rangle)(1-g)} \label{7-7}
\end{equation}
The relation between rumor spreading threshold, with inoculation ($\hat{\lambda_c}$) and without inoculation($\lambda_c$) can be defined as,
\begin{equation}
\hat{\lambda_c}=\frac{\lambda_c}{1-g} \label{7-8}
\end{equation}
It is to note that by applying random inoculation, the rumor spreading threshold ($ \hat{\lambda_c}$) can be increased as seeen in Eq. \eqref{7-8} (i.e., $ \hat{\lambda_c} >\lambda_c$). 
 
 \section{Targeted inoculation}
  Scale free networks permit efficient strategies and depend upon the hierarchy of nodes. It has been shown that SF networks shows robustness against random inoculation. It shows that the high fraction of inoculation of nodes can be resisted without loosing its global connectivity. But on the other hand SF networks are strongly affected by targeted inoculation of nodes. The SF network suffers an interesting reduction of its robustness to carry information. In targeted inoculation, the high degree nodes have been inoculated progressively, i.e more likely to spread the information. In SF networks, the robustness of the network decreases at the affect of a tiny fraction of inoculated individuals.
  
   Let us assume that fraction $ g_k $ of nodes with degree k are successfully inoculated. An upper threshold of degree $k_t$ , such that all nodes with degree $k>k_t$ get inoculated.  Fraction $ g_k $ of nodes with the degree \textit{k} are successfully inoculated. The fraction of inoculated nodes  given by,
  
   \begin{equation}
\hspace*{0.8cm}g_k\ =\ %\left \{
  \begin{cases}%{array}{cc}
   1, &  \ k>k_t,\\
   f, &  \ k=k_t, \\
   0,& \ k<k_t.
    \end{cases}\label{8-1}
    \end{equation}
    where $ 0<f\leq1 $, and $\sum_k g_kP(k)=\bar{g}$, where $ \bar{g} $ is the average inoculation fraction. Therefore, now rumor spreading equation is defined for targeted inoculation as,
    \begin{eqnarray}
 &&\frac{d\rho^i(k,t)}{dt}=-\frac{\lambda(1-g_k) k^{1+\beta}}{\langle k^{1+\beta}\rangle}\rho^i(k,t)\Phi(t);\\\label{8-2}
&&\frac{d\rho^s(k,t)}{dt}=-\rho^s(k,t) + \frac{\lambda(1-g_k) k^{1+\beta}}{\langle k^{1+\beta}\rangle}\rho^i(k,t)\Phi(t);\\\label{8-3}
&&\frac{d\rho^r(k,t)}{dt}=\rho^s(k,t).\label{8-4}
\end{eqnarray}
Further, $ \Psi $ for targeted inoculation,
\begin{eqnarray}
 \Psi(t)&=&\langle k^\alpha \rangle-\sum_{k}k^\alpha P(k`)e^{\frac{-\lambda(1-g_k) k^{1+\beta}}{\langle k^{1+\beta}\rangle}\Psi}\label{8-5}
 \end{eqnarray}
  \begin{equation}
\frac{d}{d\Psi}(\langle k^\alpha \rangle-\sum_k k^\alpha P(k)e^{\frac{-\lambda (1-g_k)k^{1+\beta}}{\langle k^{1+\beta}\rangle}\Psi})|_{\Psi=0} >1 \label{8-6}
\end{equation} 
Therefore, rumor spreading threshold in the case of targeted inoculation is obtained as,
\begin{equation}
\tilde{\lambda_c}= \frac{\langle k^{\beta+1}\rangle}{\langle k^{\alpha+\beta+1}\rangle-\langle g_k k^{\alpha+\beta+1}\rangle}\label{8-7}
\end{equation}
Here, $ \langle g_kk^{\alpha+\beta+1} \rangle $=$ \bar{g}\langle k^{\alpha+\beta+1} \rangle+\eta \prime $, where $ \eta \prime=\langle {(g_k-\bar{g})[\langle k^{\alpha+\beta+1} -\langle k^{\alpha+\beta+1} \rangle}]\rangle $ is the covariance of $ g_k $ and $k^{\alpha+\beta+1} $. The cut-off degree $ k_c $ is large enough where $ \eta \prime<0 $, but for small $ k_c $, $ g_k-\bar{g} $ and $k^{\alpha+\beta+1}- \langle k^{\alpha+\beta+1}\rangle $ have the same signs except for k's where $ g_k-\bar{g} $ and/or $ k^{\alpha+\beta+1} -\langle k^{\alpha+\beta+1} \rangle $ is 0. 

Therefore $ \eta \prime>0 $ for appropriate $ k_c$,
\begin{equation}
  \tilde{\lambda_c}>\frac{1-g}{1-\bar g}\hat{\lambda_c} \label{8-8}
\end{equation} 
 If fractions of inoculation in random and targeted are same then $ g=\bar{g} $,
 \begin{equation}
 \tilde{\lambda_c}>\hat{\lambda_c} \label{8-9}
 \end{equation}
 The above relation shows that in scale free networks targeted inoculation is more effective than the random inoculation.
 
\section{Numerical simulations: results and discussion } 
The studies of uncorrelated networks have been performed using the degree distribution of scale free network. The size of the network is considered to be $ N=10^5 $ and the degree exponent ($ \gamma $)=2.4. At the starting of rumor spreading, the spreaders are randomly chosen. In Fig. \ref{R_lam1}, the  final size of rumor \textit{R} is plotted against rumor transmission rate for \textit{N}=100000, 1000 and 100 by tuning $ \alpha $ and $ \beta $ as,

\begin{itemize}
\item$ \alpha+\beta =0 $: Finite rumor threshold has been found. It has been observed that for different size of networks, constant threshold is there (after fixing the value of $ \alpha $ and $ \beta $). For the case $ \alpha+\beta+2<\gamma $, since $ \gamma=2.4 $. Therefore here it is interesting to see that finite threshold has been found which is independent from the size of network same as obtained from  Eqs.\eqref{6-15} and \eqref{6-16}.
\item $ \alpha+\beta=-1: $ The simulation results are found same as above since $ \alpha+\beta+2<\gamma $ with finite threshold and constant for any network size (for fix values of $ \alpha $ and $ \beta $).
\item $ \alpha+\beta=1 $: The results show some rumor threshold but approaches to 0 as network size increases. For this case $ \alpha+\beta+2>\gamma $, since $ \gamma=2.4 $. Therefore, the threshold approaches to 0 as network size increases. Similar results have been obtained by Eqs.\eqref{6-15}-\eqref{6-16}.
\item $ \alpha+\beta=2$: The simulation results are found same as above since $ \alpha+\beta+2>\gamma $ with  threshold approaches to 0 as size of network increases.
\end{itemize}

Using simulation results, of final size of rumor (\textit{R}), is plotted against time (\textit{t}) in Fig. \ref{R_t}. It has been observed that rumor size increases exponentially as time increases and after some time it approaches a steady state,  that will remain constant, since spreader density is 0 at that time. It has also been observed that when $ \alpha+\beta $ is low, then rumor size increases slowly initially but when $ \alpha+\beta $ increases rumor size increases rapidly against time. While tuning  $ \alpha+\beta=-1 $ to  $ \alpha+\beta =2$ rumor increments faster initially (Fig. \ref{R_t}). When ratio of $ \alpha $ and $ \beta $ is high than the rumor size is also high. This result justifies that the $ \alpha $ affects more final rumor size \textit{R} than $ \beta $It is seen from Eq. \eqref{6-5} that  when $ \alpha $ is very small (0.1. 0.3) then the rumor threshold will be high. Than it is seen the final size of rumor will be too small when rumor transmission rate ($ \lambda $) is less than the 
rumor threshold ($ \lambda_c $).
\begin{figure}[ht]
\begin{center}
$\begin{array}{cc}
\includegraphics[width=2.4in, height=2 in]{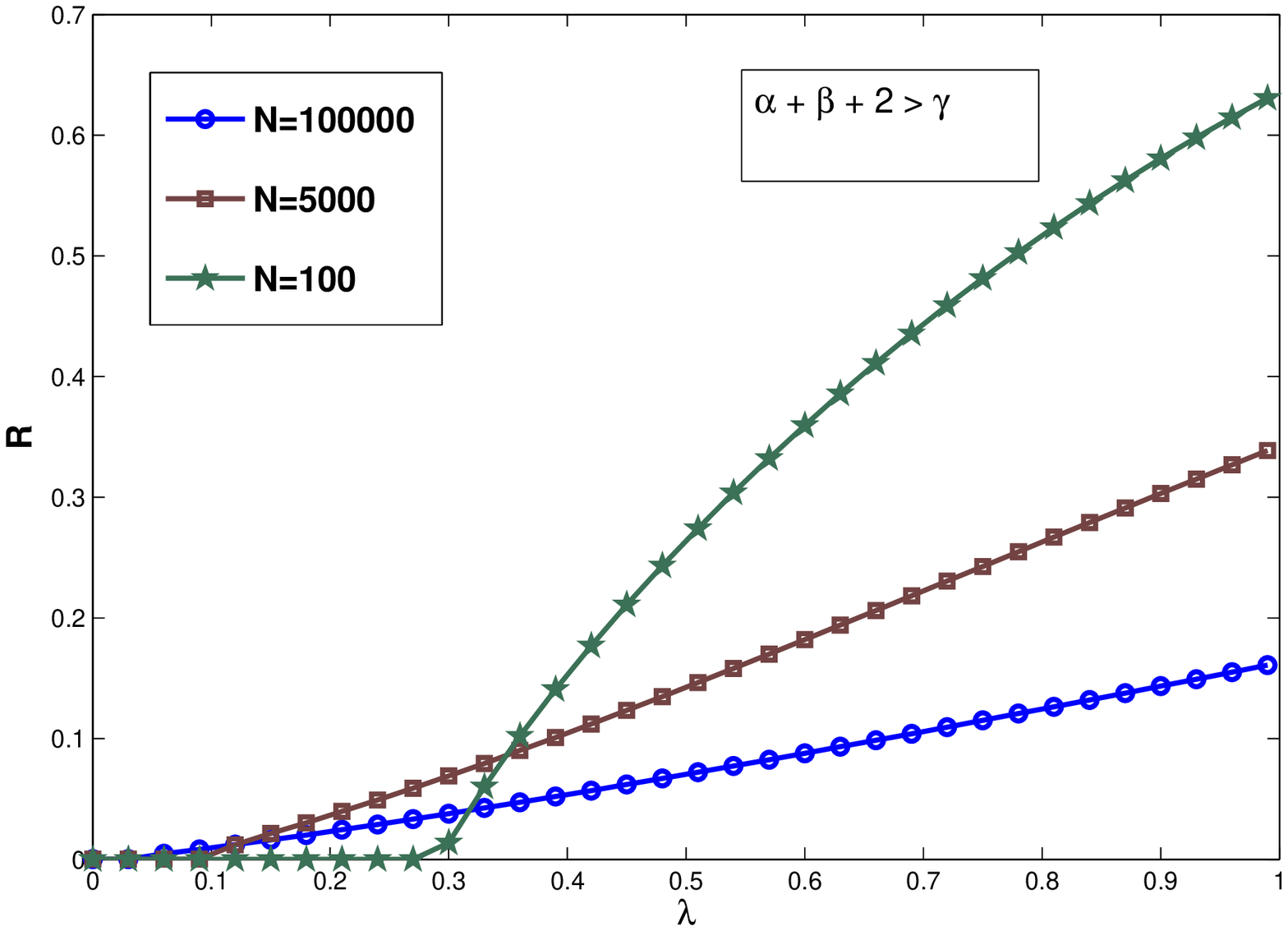} &
\includegraphics[width=2.4in, height=2 in]{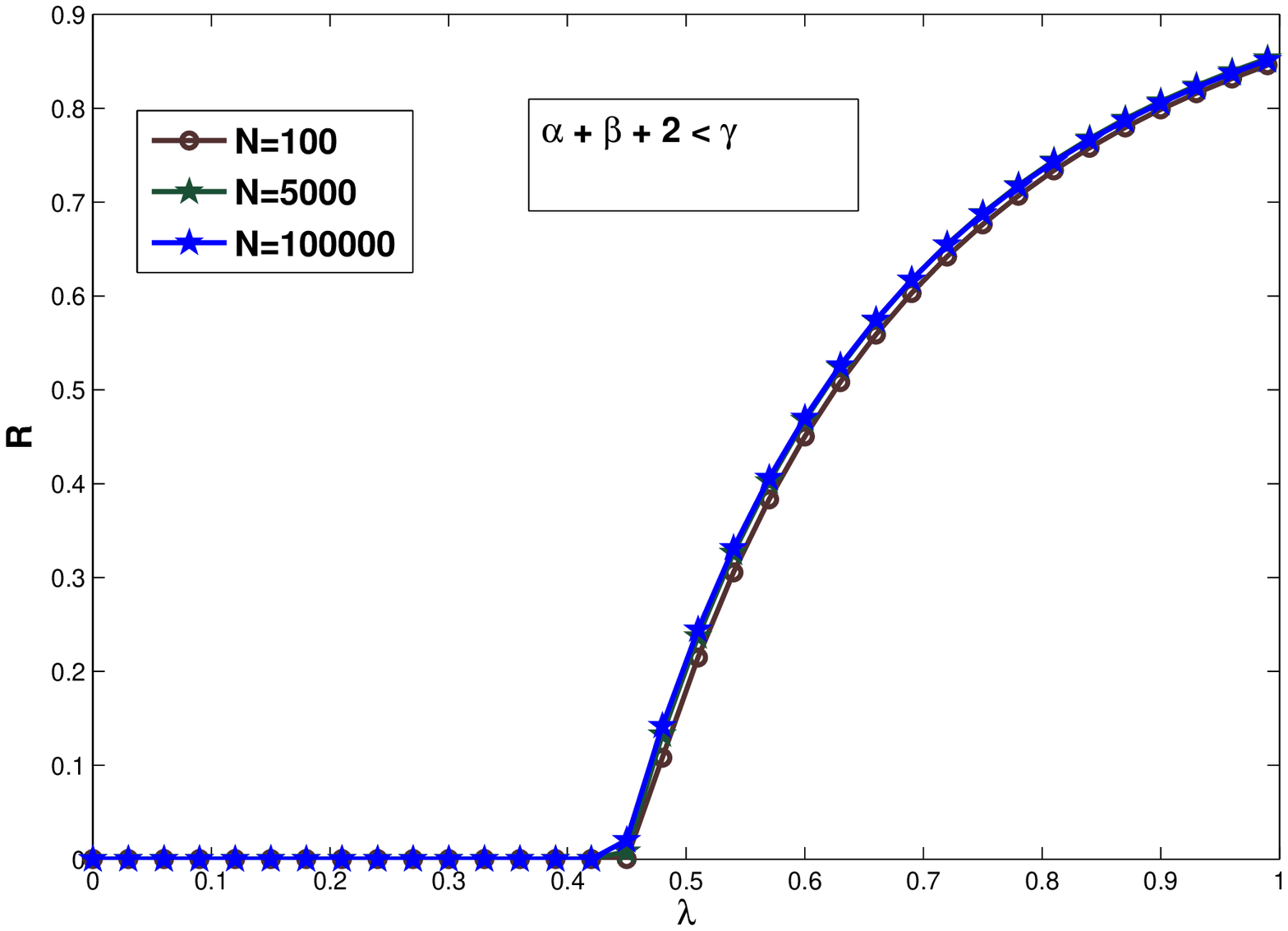}\\ 
\mbox{\textbf{(a)} R vs $\lambda$ for $ \alpha+\beta+2>\gamma $ } & \mbox{\textbf{(b)} R vs $\beta$ for $ \alpha+\beta+2<\gamma $} \\
\end{array}$
\end{center}
\caption{ \textit{R} vs $\lambda$ with $ \alpha+\beta+2>\gamma $ and $ \alpha+\beta+2<\gamma $ for different size of scale free networks } \label{R_lam1}
\end{figure}

\begin{figure}
\begin{center}
$\begin{array}{cc}
\includegraphics[width=2.7in, height=1.8 in]{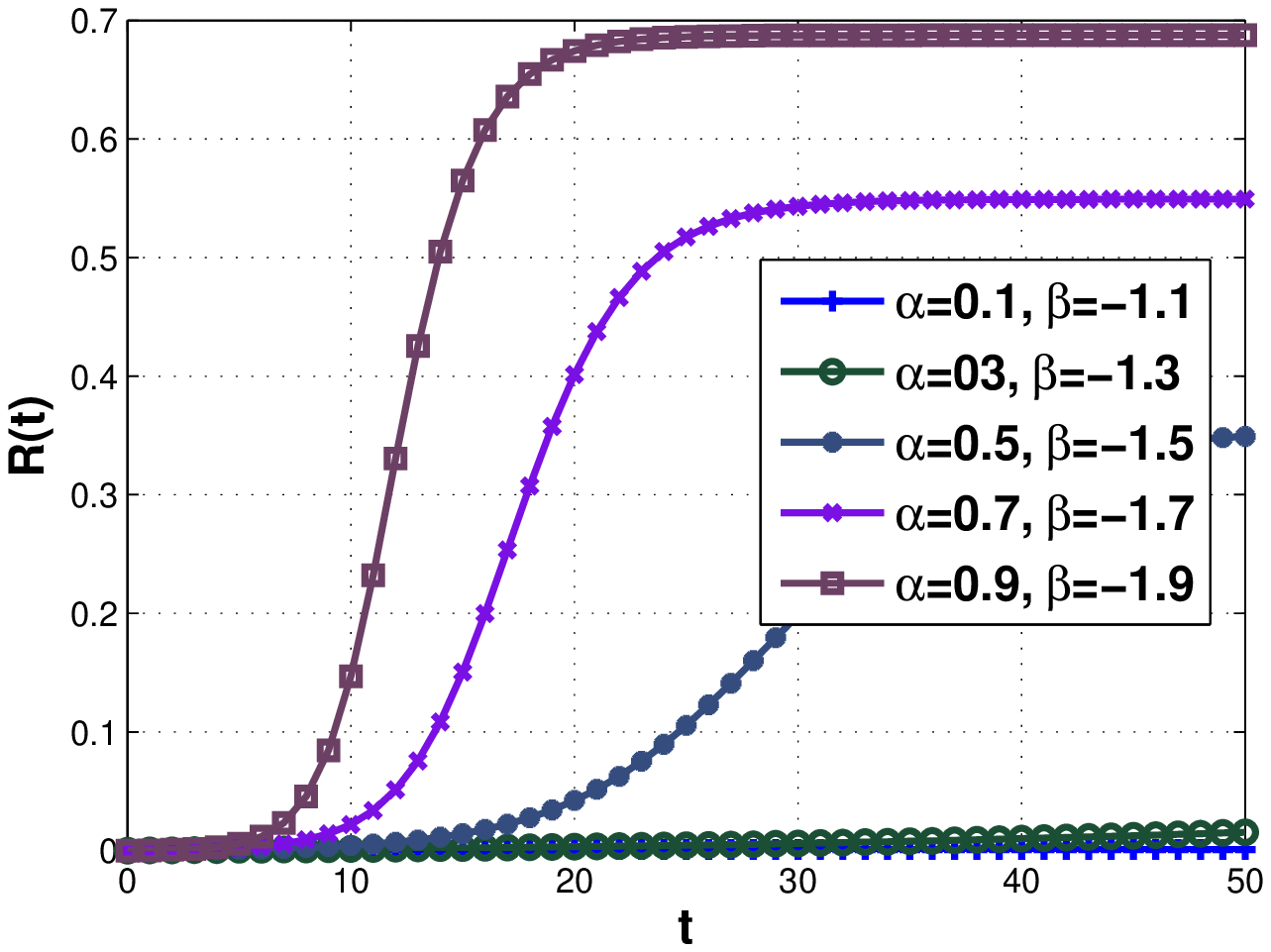} &
\includegraphics[width=2.7in, height=1.8 in]{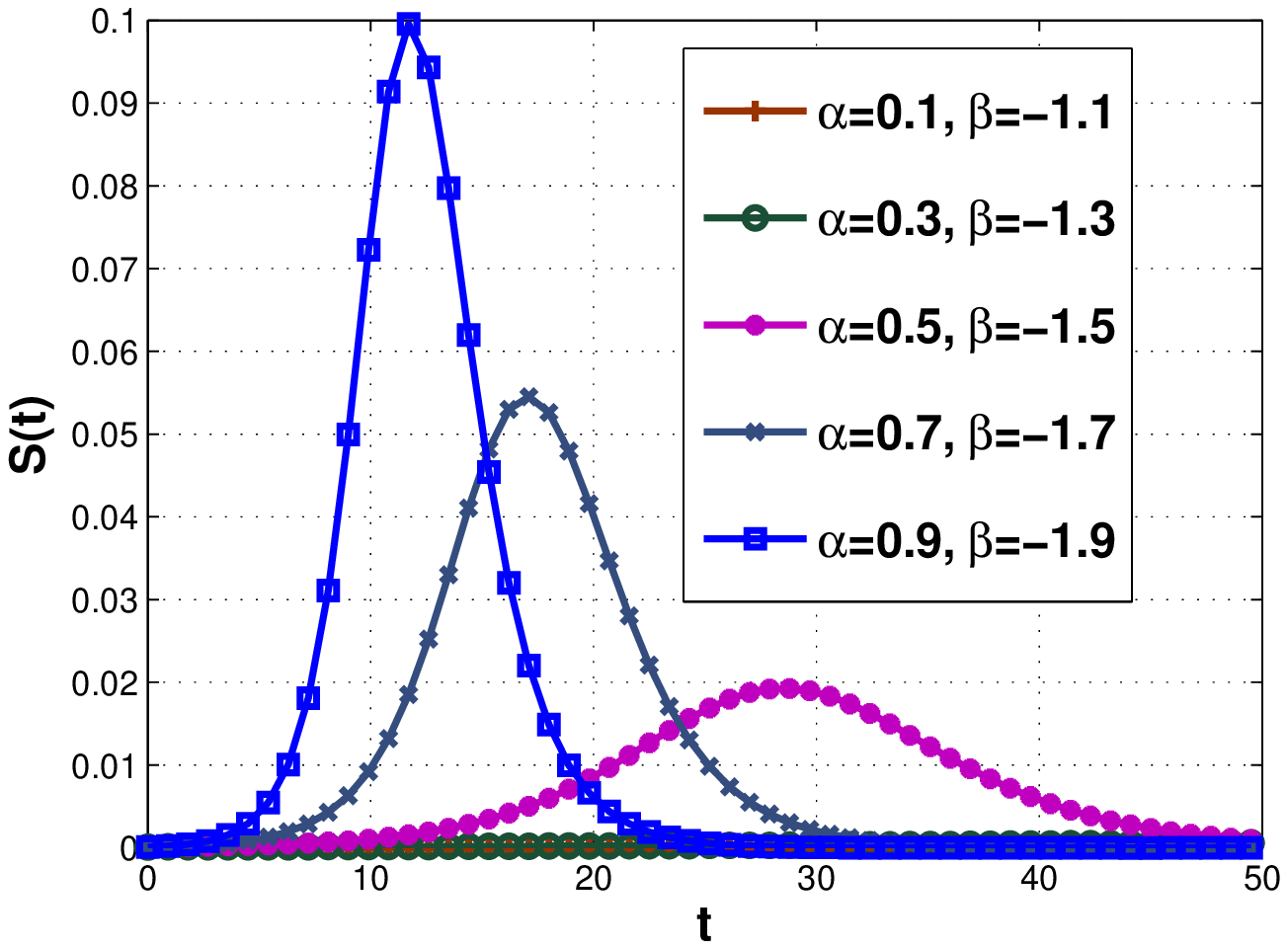}\\ 
\mbox{\textbf{(a)} $\alpha+\beta =-1$} & \mbox{\textbf{(e)} $\alpha+\beta =-1$}\\
\includegraphics[width=2.7in, height=1.8 in]{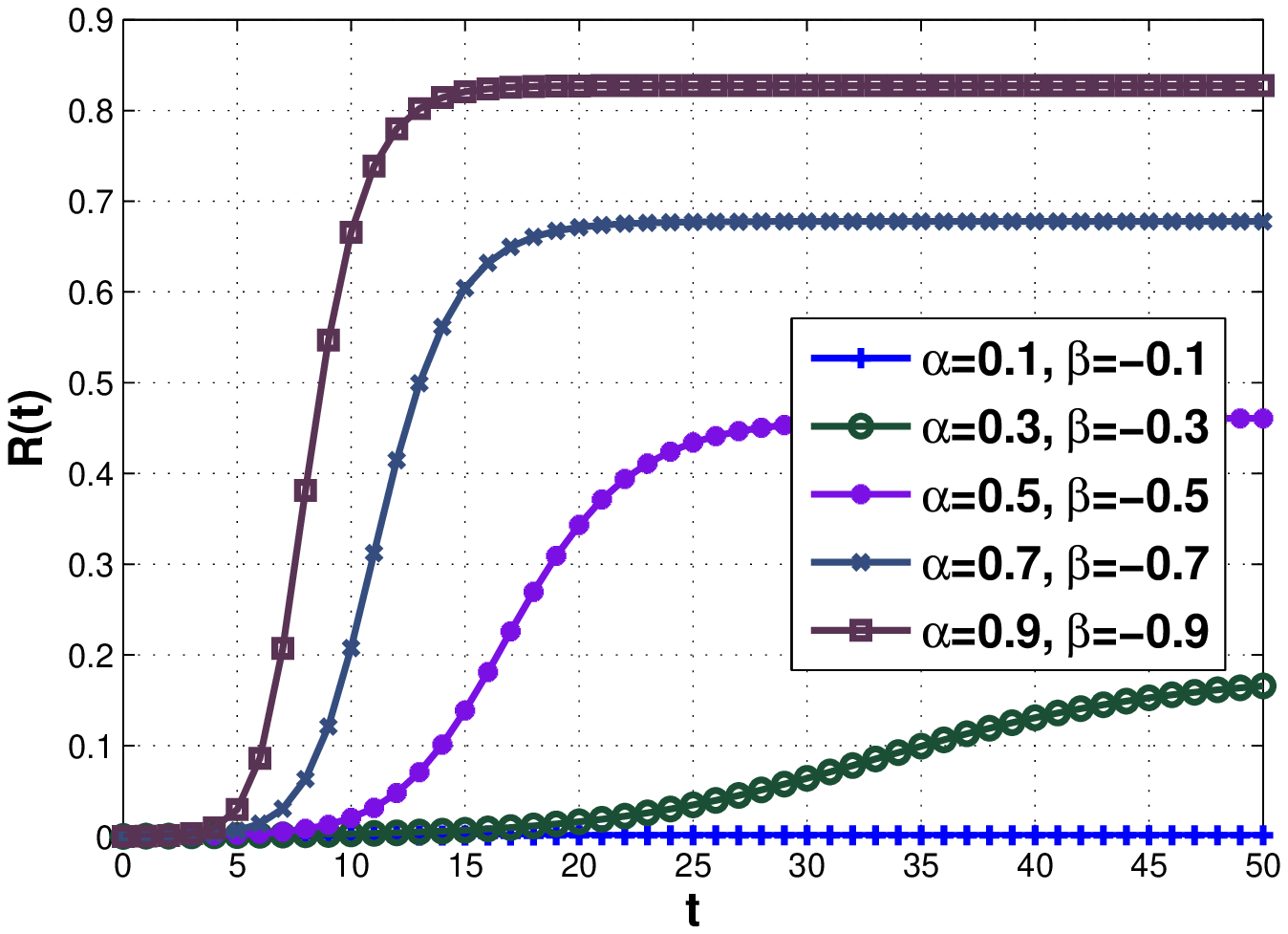} &
\includegraphics[width=2.7in, height=1.8 in]{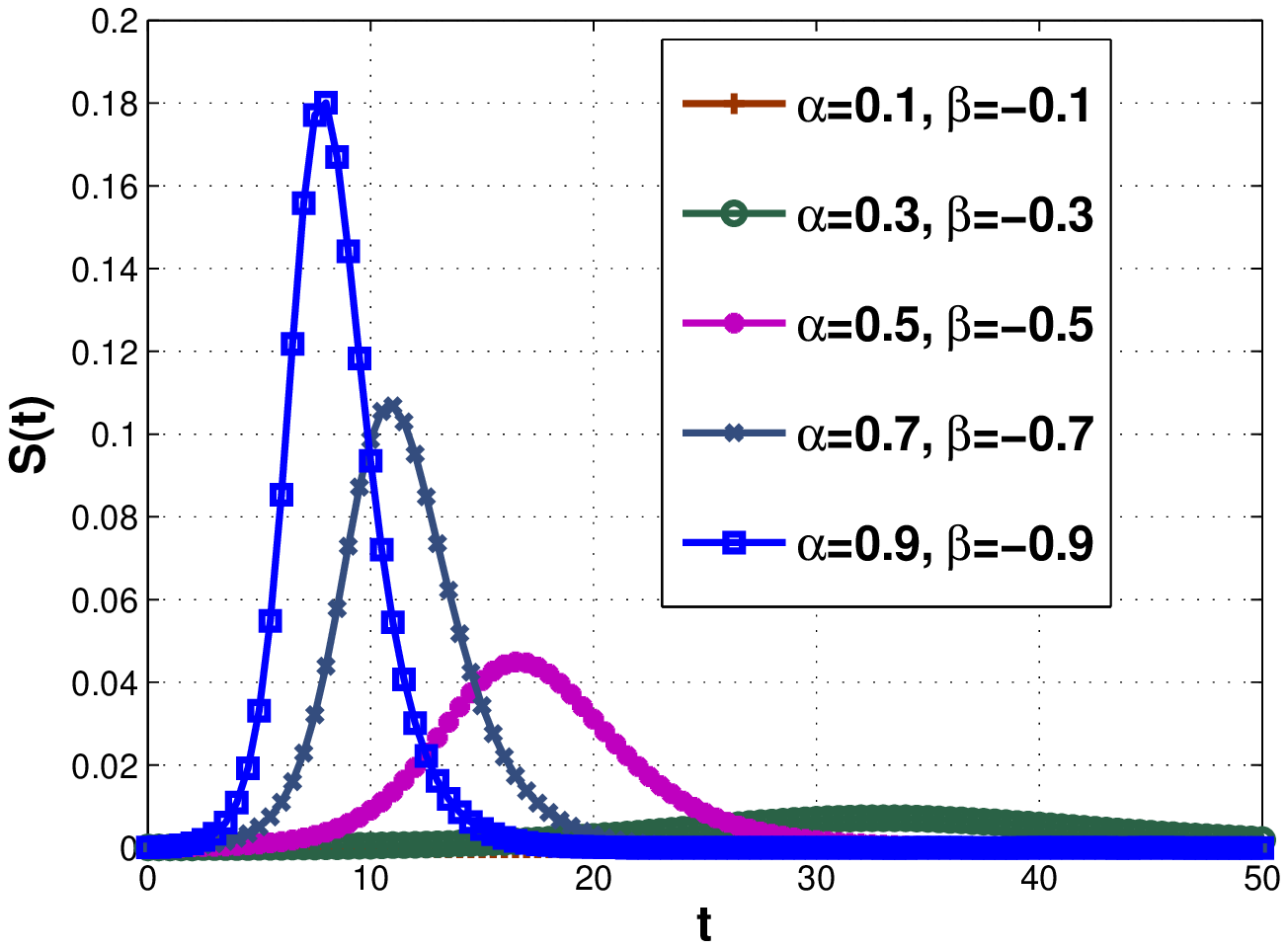} \\
\mbox{\textbf{(b)} $\alpha+\beta =0$} & \mbox{\textbf{(f)} $\alpha+\beta =0$}\\
\includegraphics[width=2.7in, height=1.8 in]{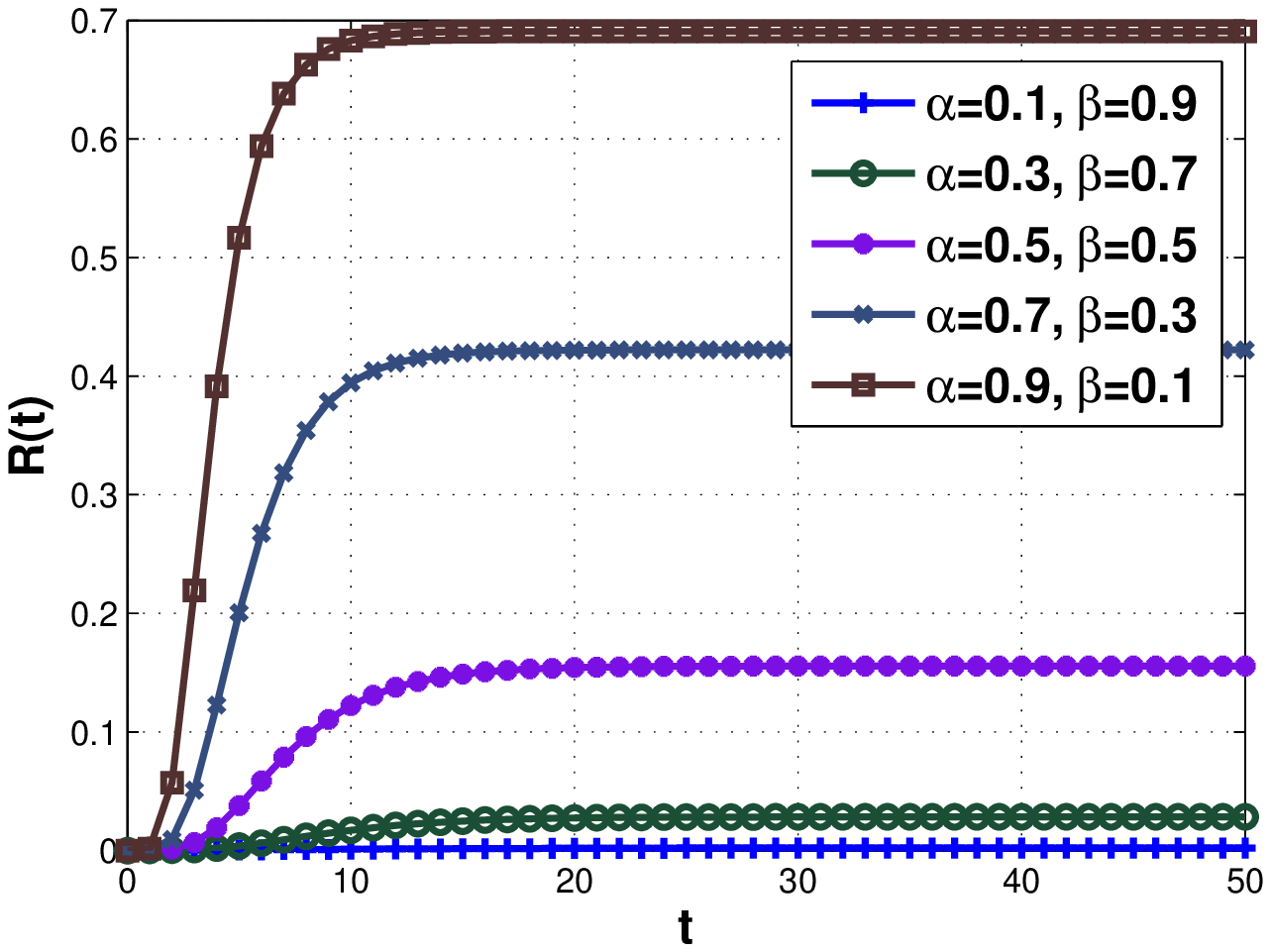} &
\includegraphics[width=2.7in, height=1.8 in]{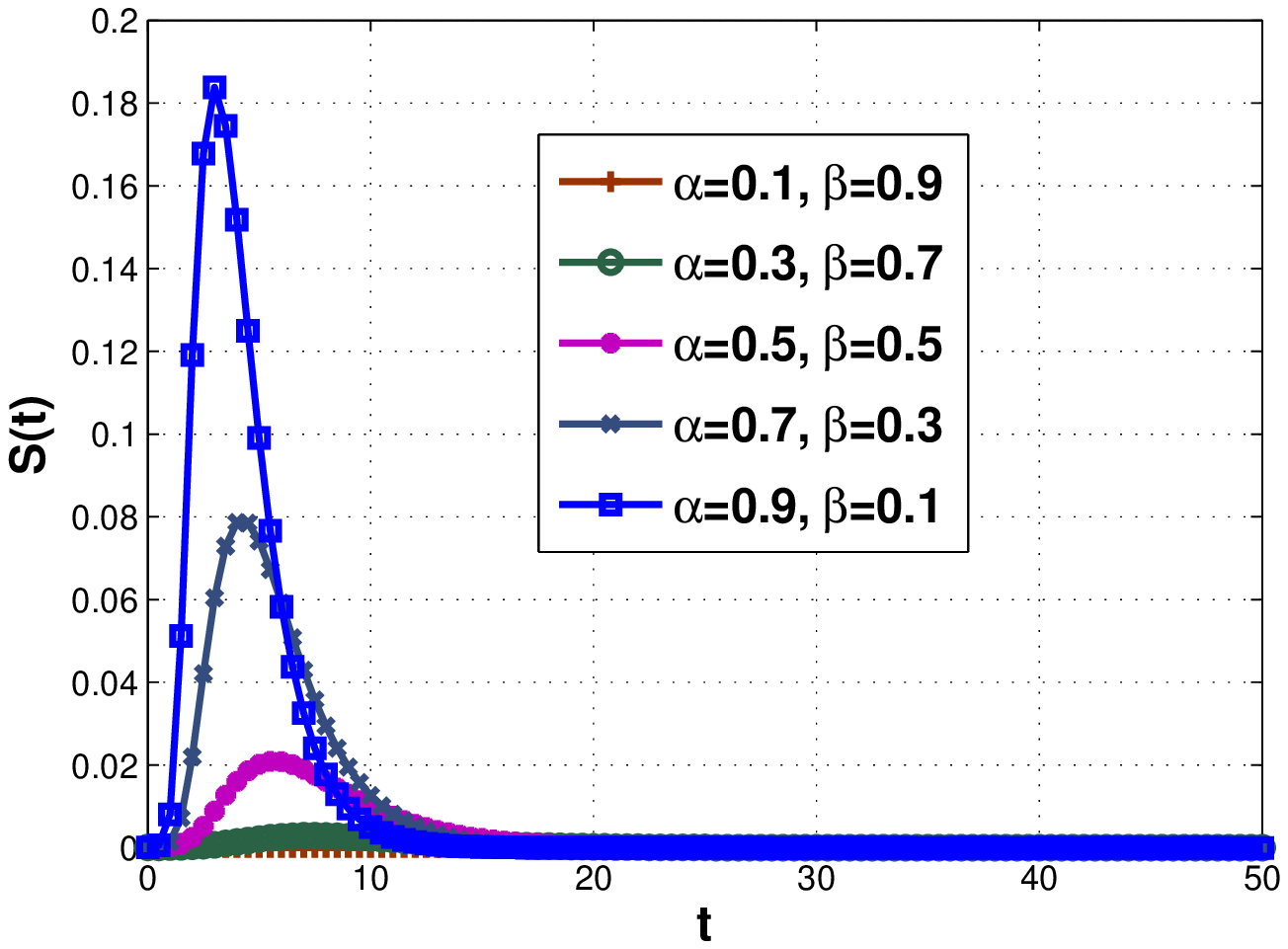} \\
\mbox{\textbf{(c)} $\alpha+\beta =1$} & \mbox{\textbf{(g)} $\alpha+\beta =1$}\\
\includegraphics[width=2.7in, height=1.8 in]{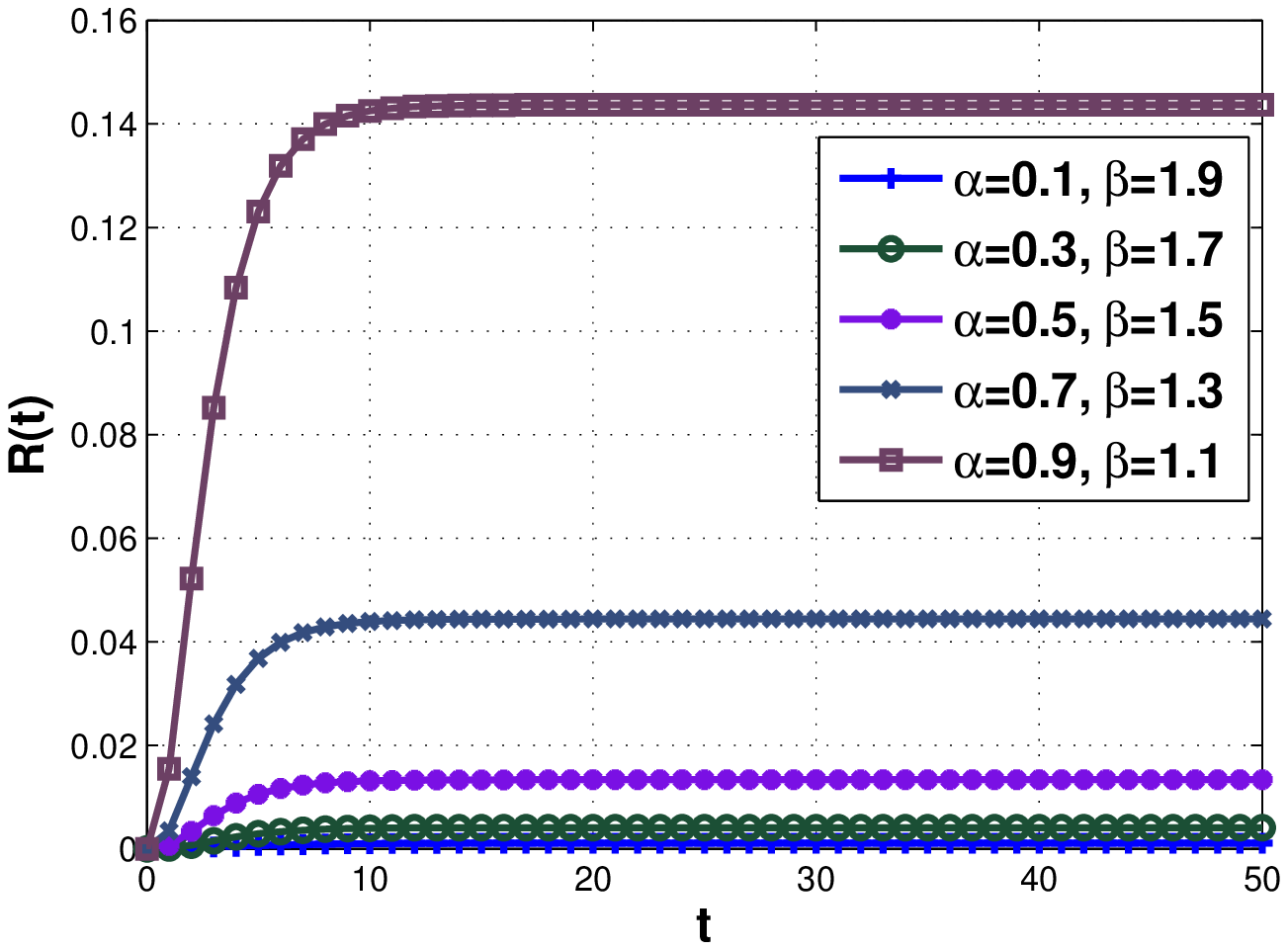} &
\includegraphics[width=2.7in, height=1.8 in]{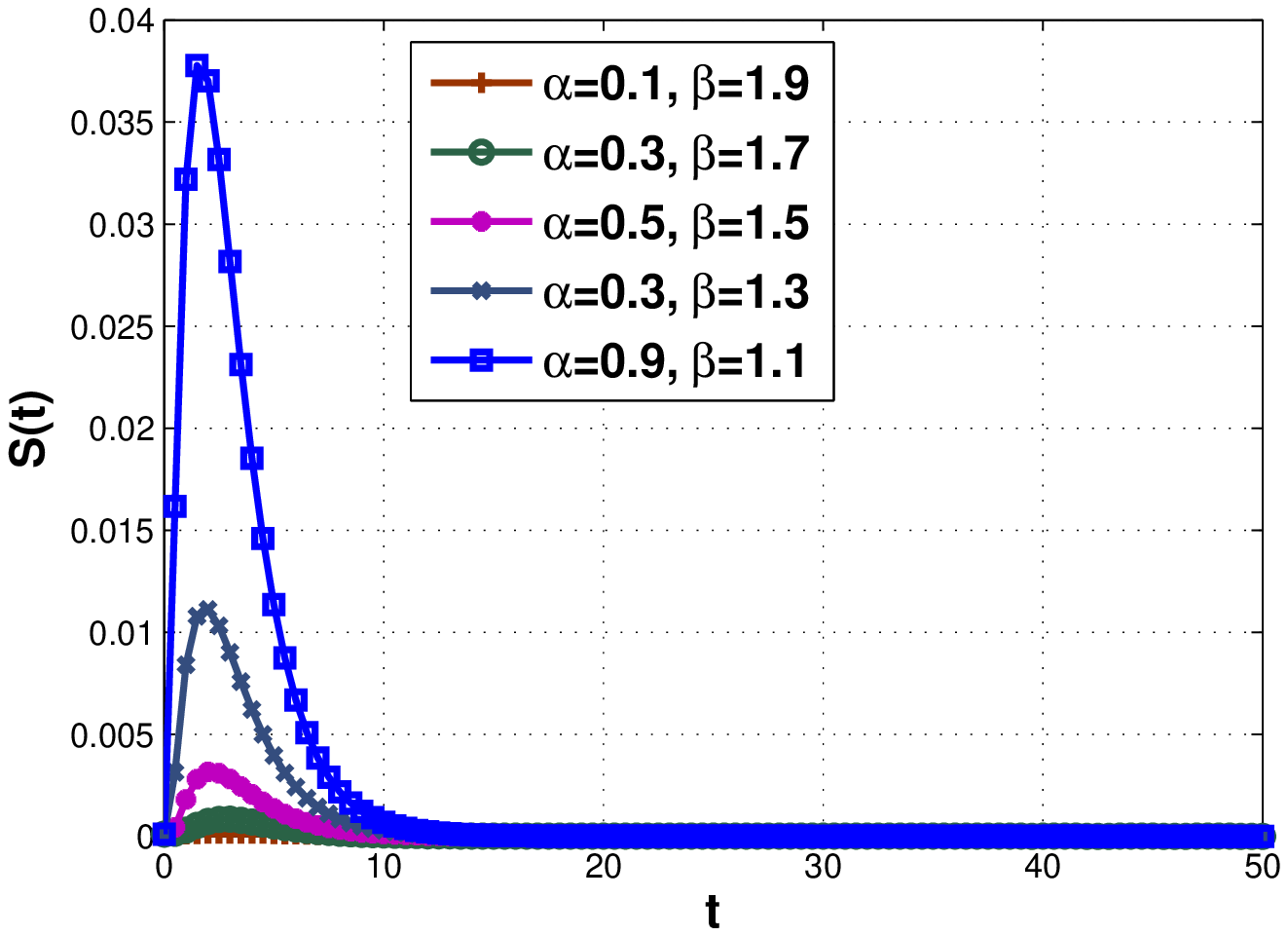} \\
\mbox{\textbf{(d)} $\alpha+\beta =2$} & \mbox{\textbf{(h)} $\alpha+\beta =2$}
\end{array}$
\end{center}
\caption{\textit{R(t)} vs \textit{t} (a-d) and \textit{S(t}) vs \textit{t} (e-h) for $ \lambda=0.8 $ and different combinations of $ \alpha $ and $ \beta $ parameters} \label{R_t}
\end{figure}

\begin{figure}[h]
\begin{center}
$\begin{array}{cc}
\includegraphics[width=2.4in, height=2 in]{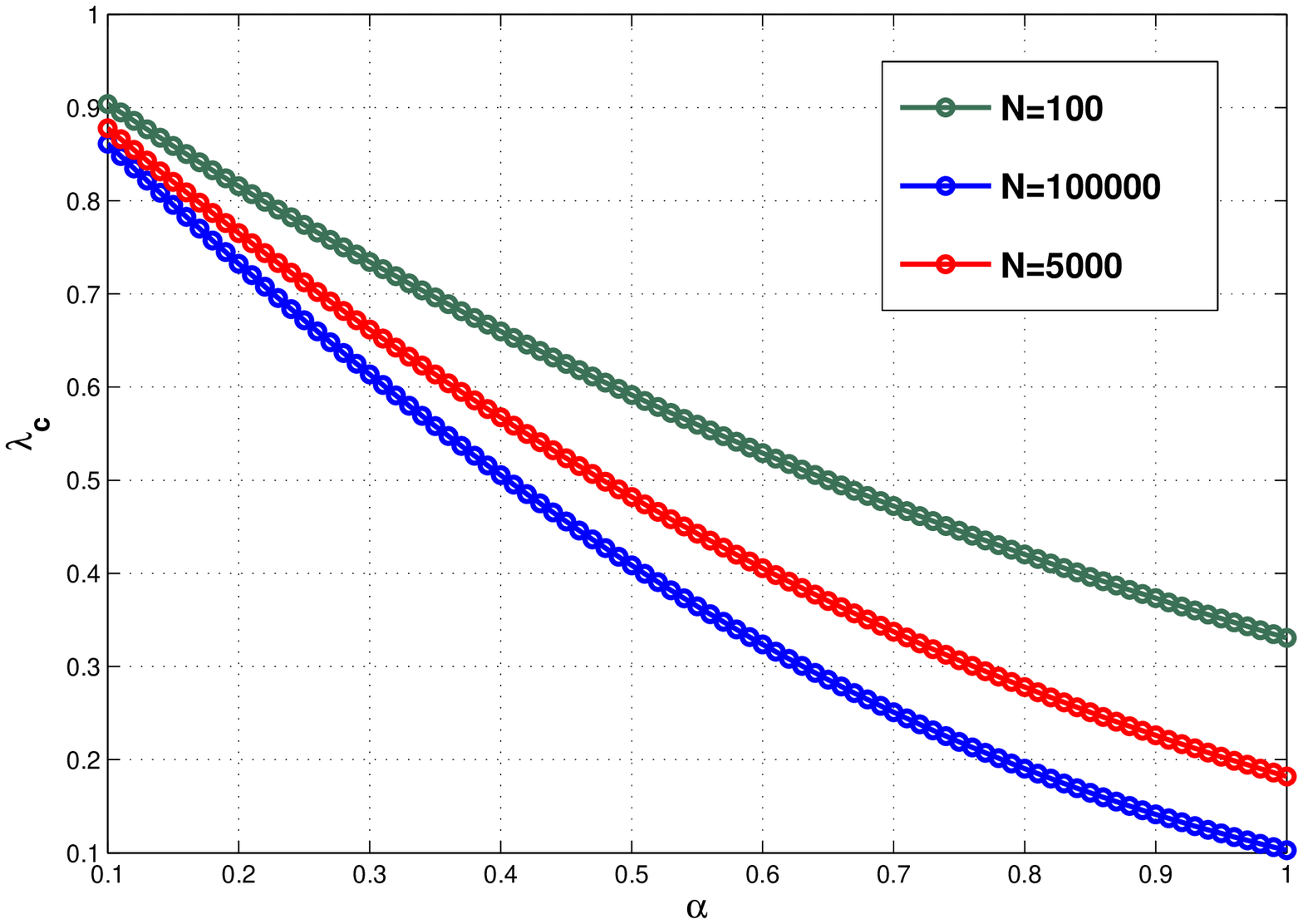} &
\includegraphics[width=2.4in, height=2 in]{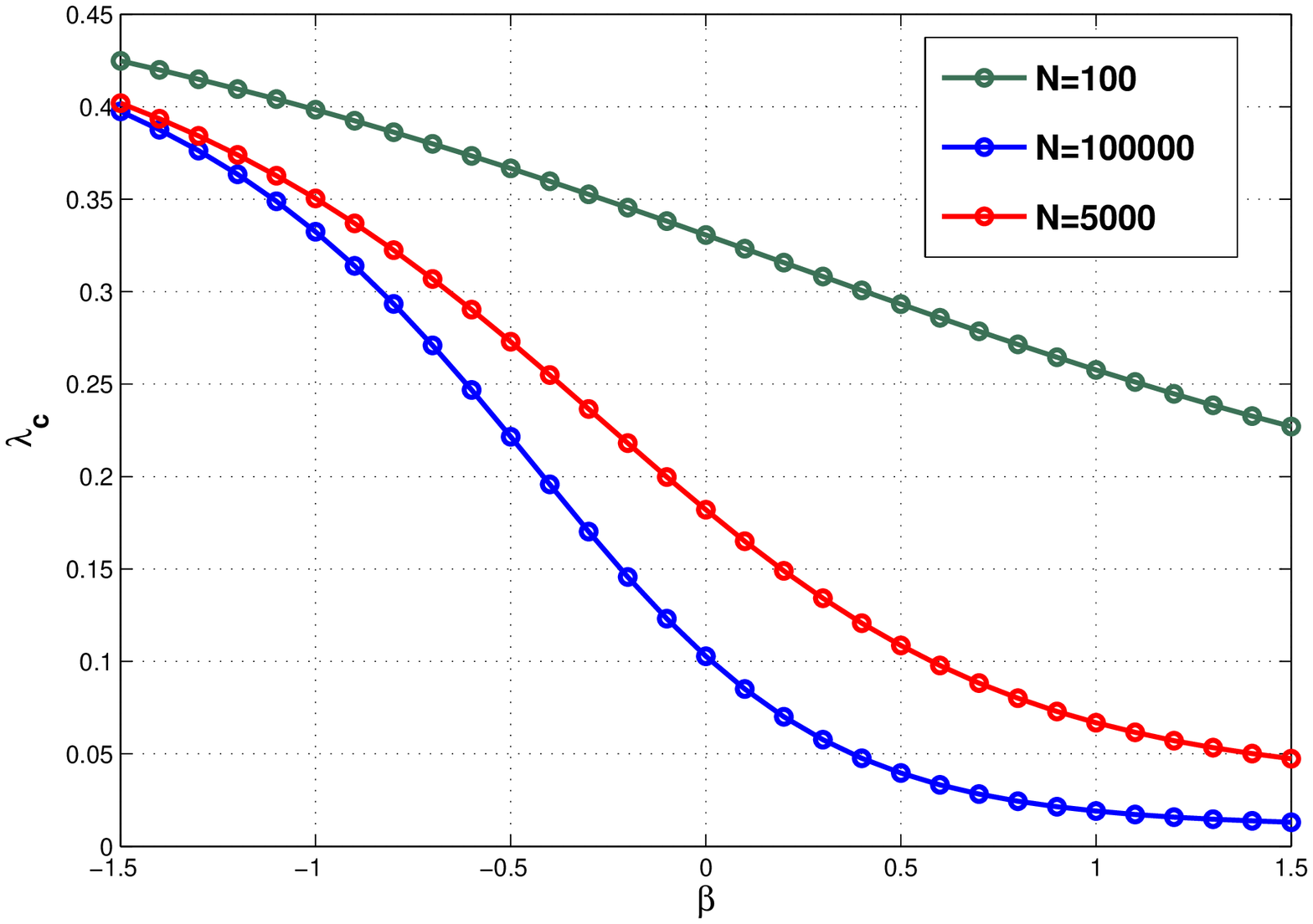}\\ 
\mbox{\textbf{(a)} $\lambda_c$ vs $\alpha$} & \mbox{\textbf{(b)} $\lambda_c $vs $ \beta$}\\
\end{array}$
\end{center}
\caption{Threshold ($ \lambda_c $) vs $ \alpha $ and $\beta  $ for $ \gamma=2.4 $ for different size of scale free networks} \label{Crc_thr}
\end{figure}
 Critical rumor threshold is plotted against $ \alpha $ in Fig. \ref{Crc_thr} while considering $ \beta=0 $. Here, $ \lambda_c $ decreased exponentially with the increase of $ \alpha $. It is maximum for $\alpha=0.1  $ and almost 0  at $ \alpha=1 $ for $ N=100000 $. Interestingly, this also happens in real  life situation, when an informed node pass information to its maximum number of neighbors then rumor spreading will get outbreak in the network. However, the outbreak is hard to achieve when it passes rumor to less number of neighbors ($\simeq$ 10-30\%). Similarly, $ \lambda_c $ has been studied against $ \beta $ at $ \alpha=1 $ in Fig. \ref{Crc_thr}. It is found that $ \beta $ affects less rumor threshold for entire range except when $ \beta>0 $. Further it  approaches to 0. When size of the network ($ N $) increases then rumor threshold is decreased as shown in Fig. \ref{Crc_thr}(a) and Fig. \ref{Crc_thr}(b).
 \begin{figure}[h]
\begin{center}
$\begin{array}{cc}
\includegraphics[width=2.5in, height=2in]{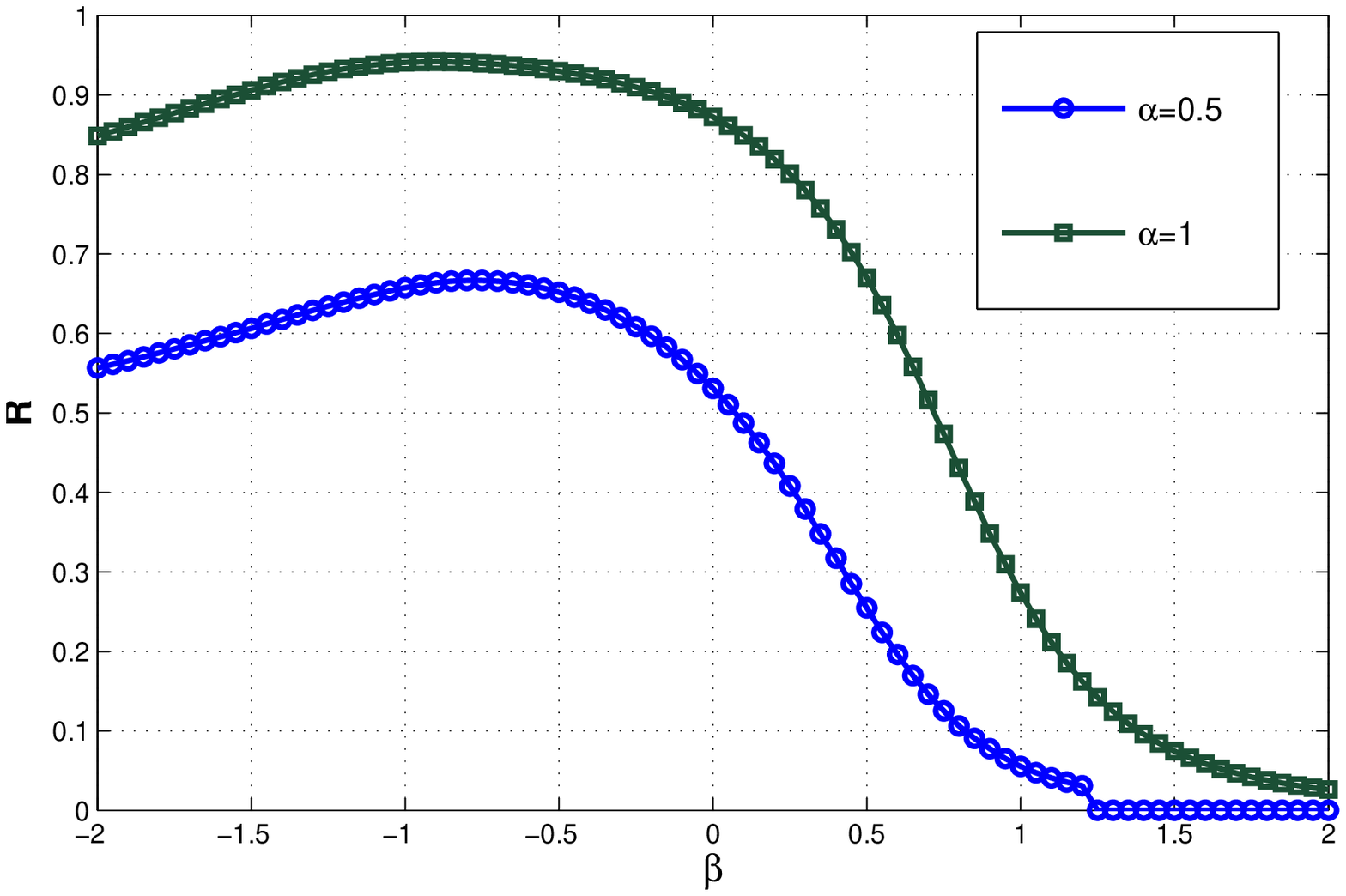} &
\includegraphics[width=2.7in, height=2.15in]{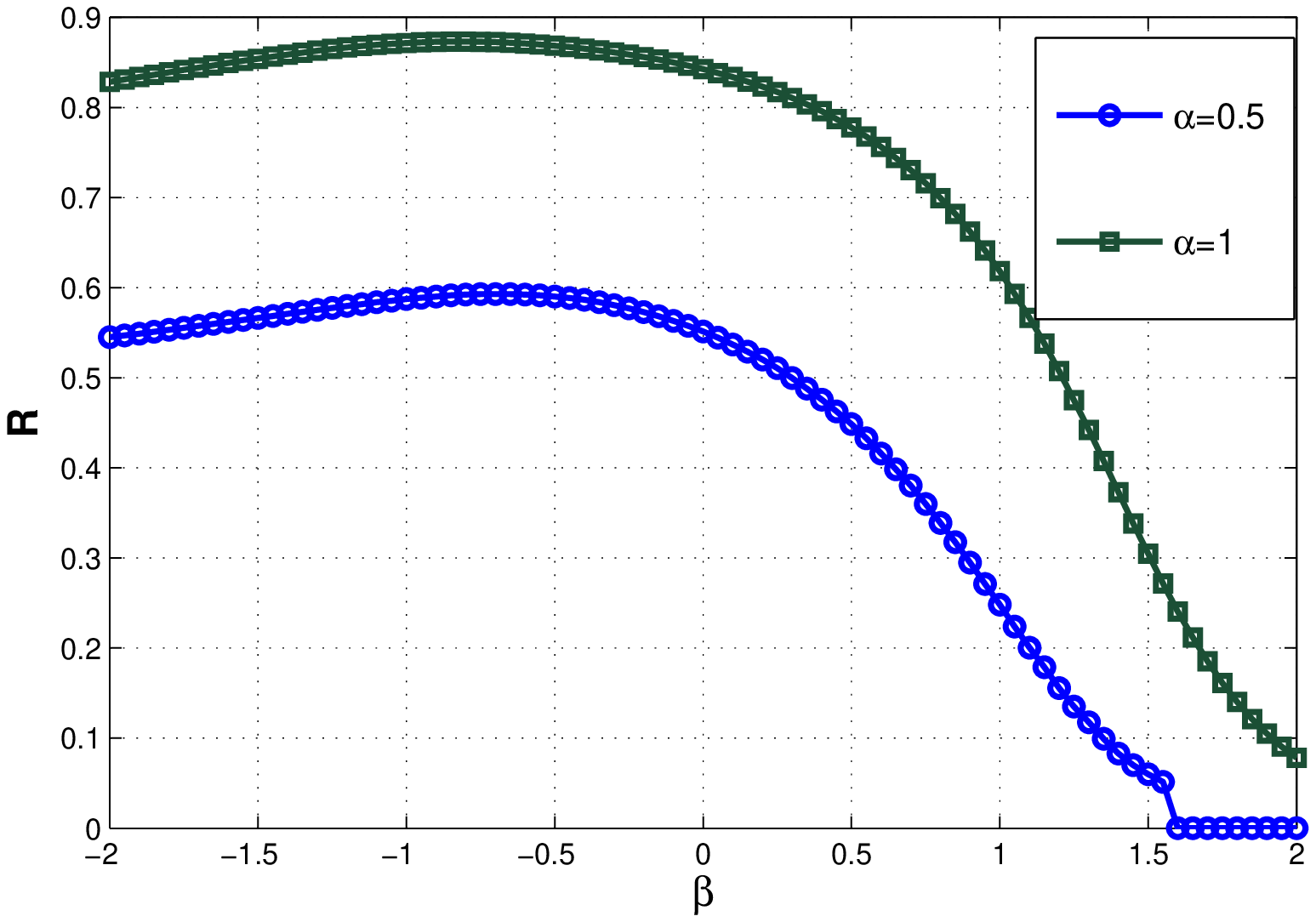}\\ 
\mbox{\textbf{(a)} $ \gamma=2.4 $} & \mbox{\textbf{(b)} $ \gamma=3 $}
\end{array}$
\end{center}
\caption{Final size of rumor \textit{R} vs $\beta$ ($ N=10^5$ nodes, $\alpha=0.5 $ and $ \alpha=1 $) for different values of $ \gamma $}\label{R_bet}
\end{figure}

\begin{figure}[h]
\begin{center}
\includegraphics[width=2.5in, height=2in]{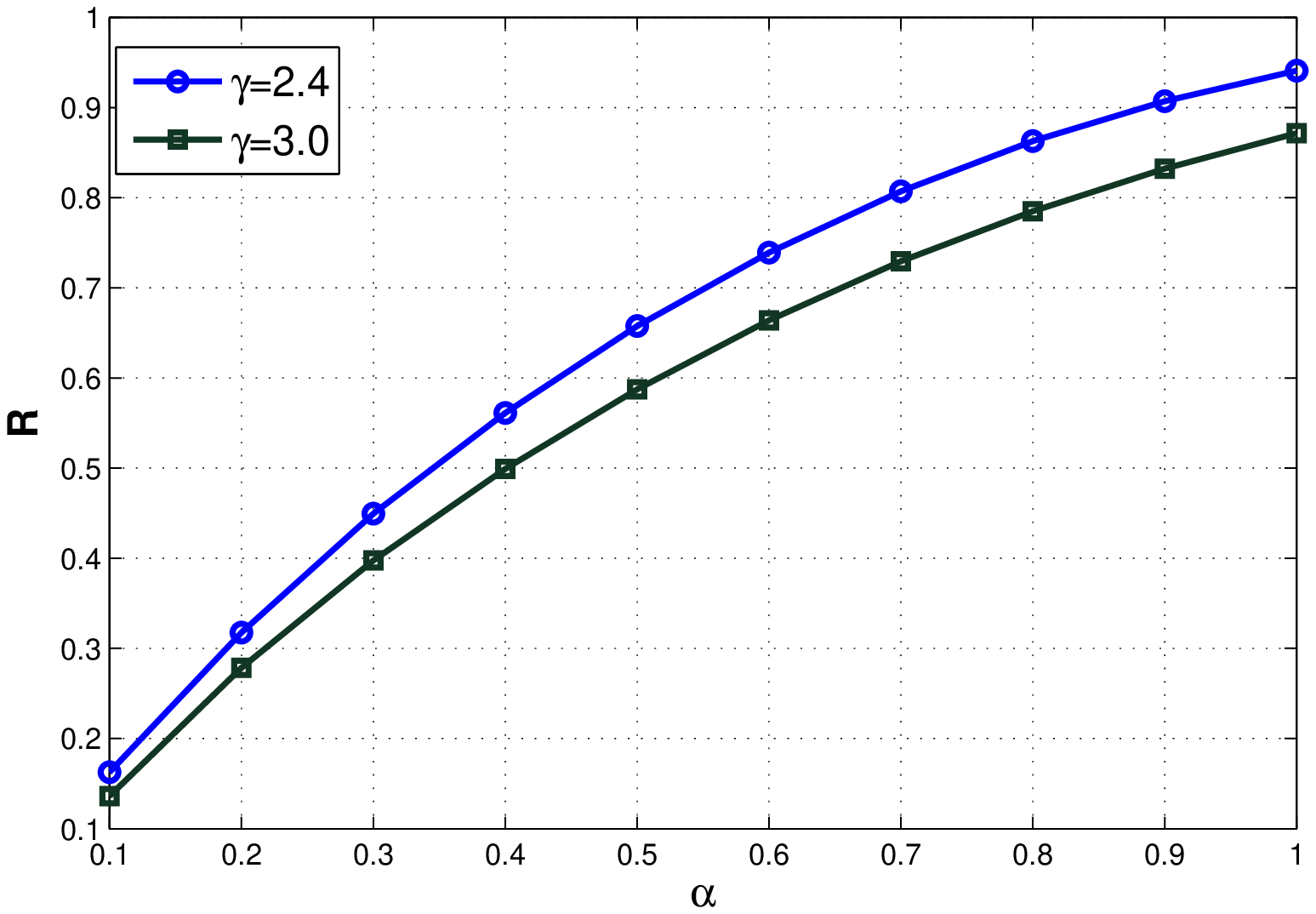} \\
\end{center}
\caption{Final size of rumor \textit{R} vs $\alpha$ ($ N=10^5$ nodes, $\beta=-1 $ and $ \lambda=1 $) for $ \gamma =2.4, 3$} \label{R_alp}
\end{figure}

In Fig. \ref{R_bet}, final size of rumor has been plotted against $ \beta $ for $ \alpha=1$, $ \gamma=2.4 $, 3, and $ \lambda=1 $. It is seen that that \textit{R} is maximum when $ \beta=-1 $. Initially rumor size \textit{R } increases with $ \beta $ but after achieving a maximum value for $ \beta=-1 $ it decays exponentially. Further, for $ \alpha=1 $, $ \gamma=2.4 $ final rumor size \textit{R} approaches to 0 (beyond $ \beta=1.5$). Furthermore it is interesting to note that, final rumor size \textit{R} increased with increase of $ \alpha $ same as obtained by Eq. \eqref{6-5} (Fig. \ref{R_alp}). For random inoculation $ g= 0.1, 0.3, 0.5, 0.7, 0.9 $, the final rumor size \textit{R} has been plotted against $ \beta $ (Fig. \ref{R_beta_rand}). It is observed that to get maximum value of \textit{R}, $ \beta $ increases when \textit{g} increases. Also, maximum size of rumor decreases with increase of \textit{g}. It is because in random inoculation rumor threshold value is larger than the threshold in model 
without inoculation as inferred from from Eq. \eqref{7-8} ($ \hat{\lambda_c}>\lambda_c $). The sharp decrease in the value of \textit{R} is seen when rumor transmission rate ($ \lambda $) is decrease by 0.5 in comparison to the case where decrease of \textit{R} is shallow when $ \alpha $ is decrease by 0.5. Since for $ \lambda \simeq 1 $ there may be a chance that rumor will spread to some extent at any value of $\textit{g} < 1,$ for very large \textit{N}. Similar results have been observed in the case of targeted inoculation in Fig. \ref{R_beta_targ}. But here maximum rumor size is much smaller with the inoculation of very less fraction of nodes (e.g. for \textit{g}=0.25 the final rumor size \textit{R} is almost suppressed), since rumor threshold in targeted inoculation is larger than the random inoculation. For random inoculation strategy, the rumor spreading is plotted against time evolution using modified model through simulation results. It is found in Fig. \ref{Rrt1}- Fig. \ref{Rrt4}, for $ g= 0.1, 0.3,
 0.5, 0.7 $ that, if $ \alpha+\beta $ increases from -1 to 2, \textit{R} will increase since rumor threshold decreases. Further it can be observed from Eq.\eqref{6-16} that, for $  \alpha+\beta=-1 $ and 0, $ \lambda_c $ is finite and higher than the case where $  \alpha+\beta =1$ and 2. Therefore, \textit{R(t)} is almost 0 and grows slowly with time when, $  \alpha+\beta=-1 $. The growth in \textit{R(t)} is higher for lower values of \textit{g}, but the case is reversed for the higher values of \textit{g}. Similarly in the case of targeted inoculation scheme using lower values of $g= 0.05, 0.1, 0.15, 0.2$ for $ \alpha+\beta=-1 $ to 2, rumor threshold found more than the random inoculation scheme (Eq. \eqref{8-9}) and rumor spreading is supressed for inoculation less number of nodes than the random inoculation scheme (Fig. \ref{Rtt1}- Fig. \ref{Rtt4}).
 \begin{figure}[h]
\begin{center}
$\begin{array}{ccc}
\includegraphics[width=1.8in, height=1.4 in]{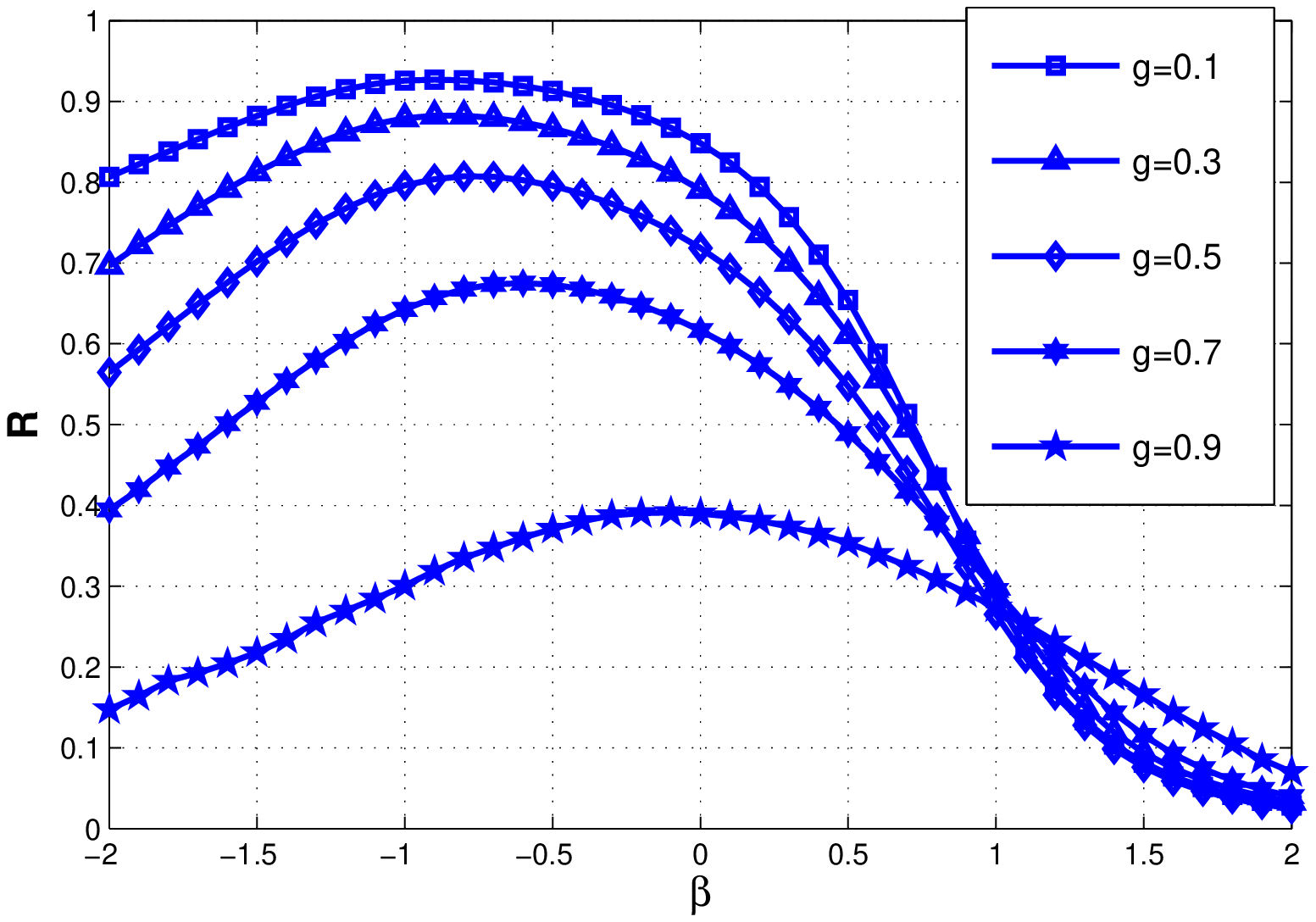} &
\includegraphics[width=1.8in, height=1.4 in]{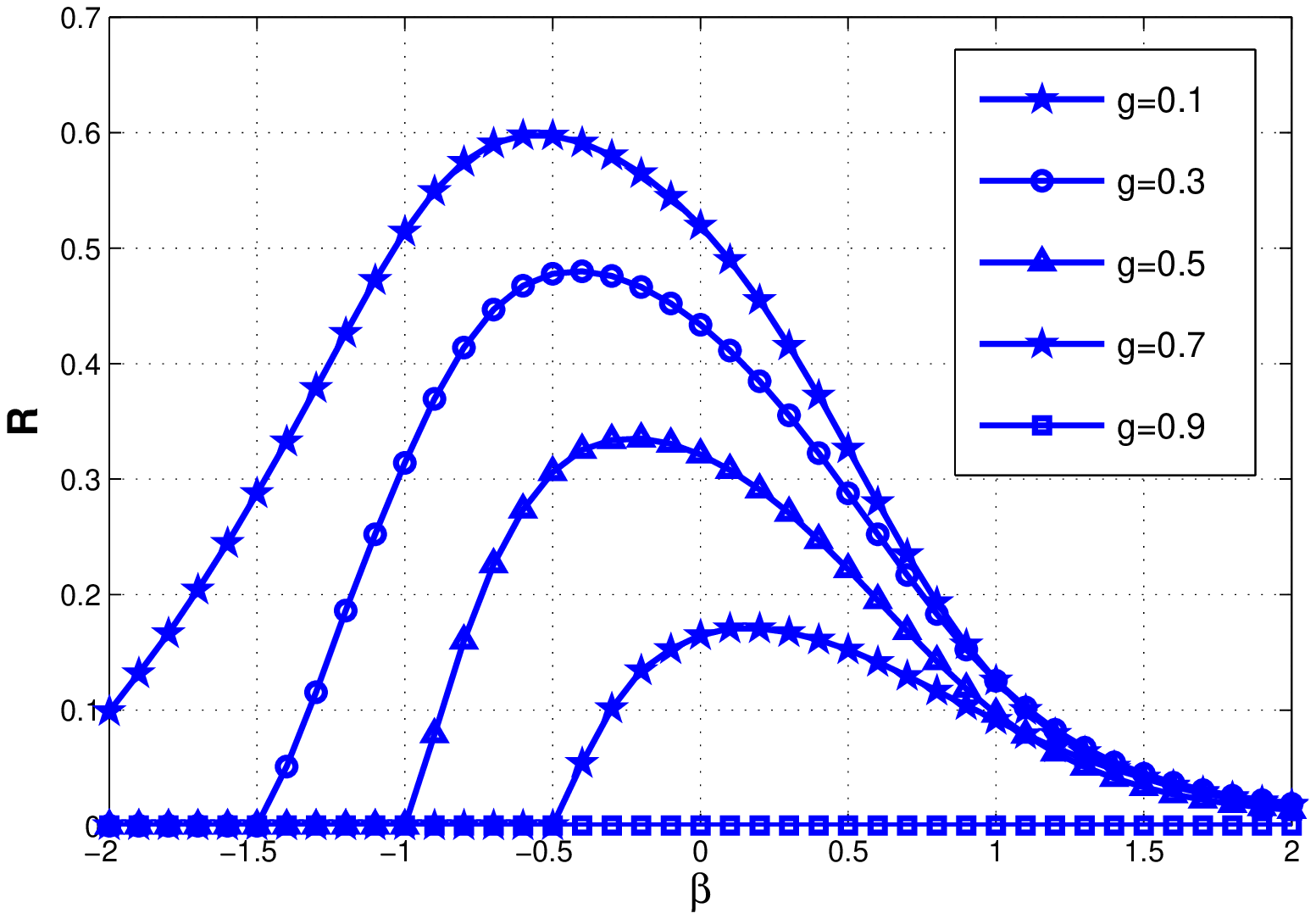} &
\includegraphics[width=1.8in, height=1.4 in]{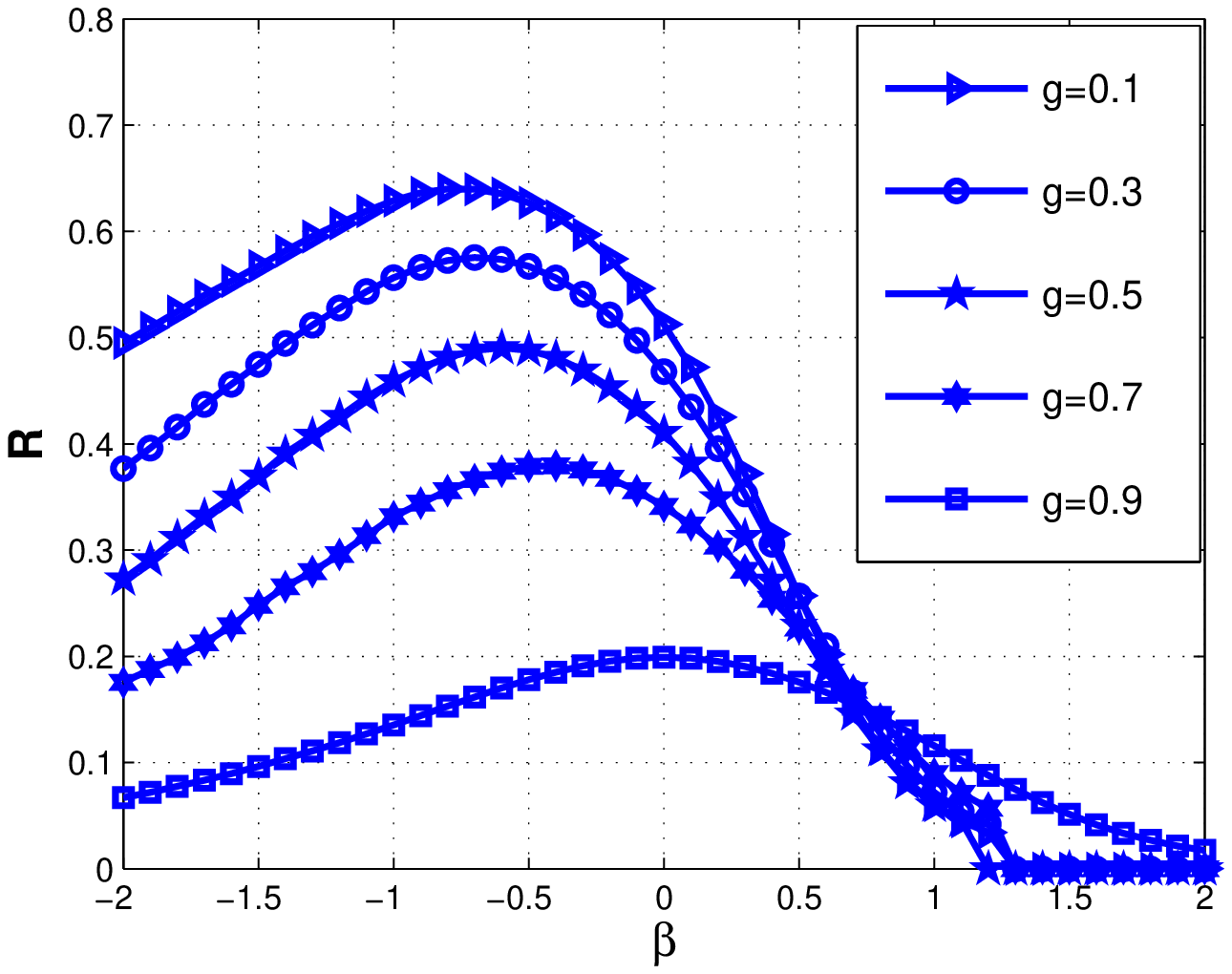}\\
\mbox{\textbf{(a)} $ \alpha=1 $ and $ \lambda=1 $} & \mbox{\textbf{(b)}  $ \alpha=1 $ and $ \lambda=0.5 $} & \mbox{\textbf{(c)}  $ \alpha=0.5 $ and $ \lambda=1 $}\\
\end{array}$
\end{center}
\caption{Final size of rumor \textit{R} vs $ \beta $ in random inoculation scheme for different fraction of inoculation (\textit{g})} \label{R_beta_rand}
\end{figure}

\begin{figure}[h]
\begin{center}
$\begin{array}{ccc}
\includegraphics[width=1.8in, height=1.4 in]{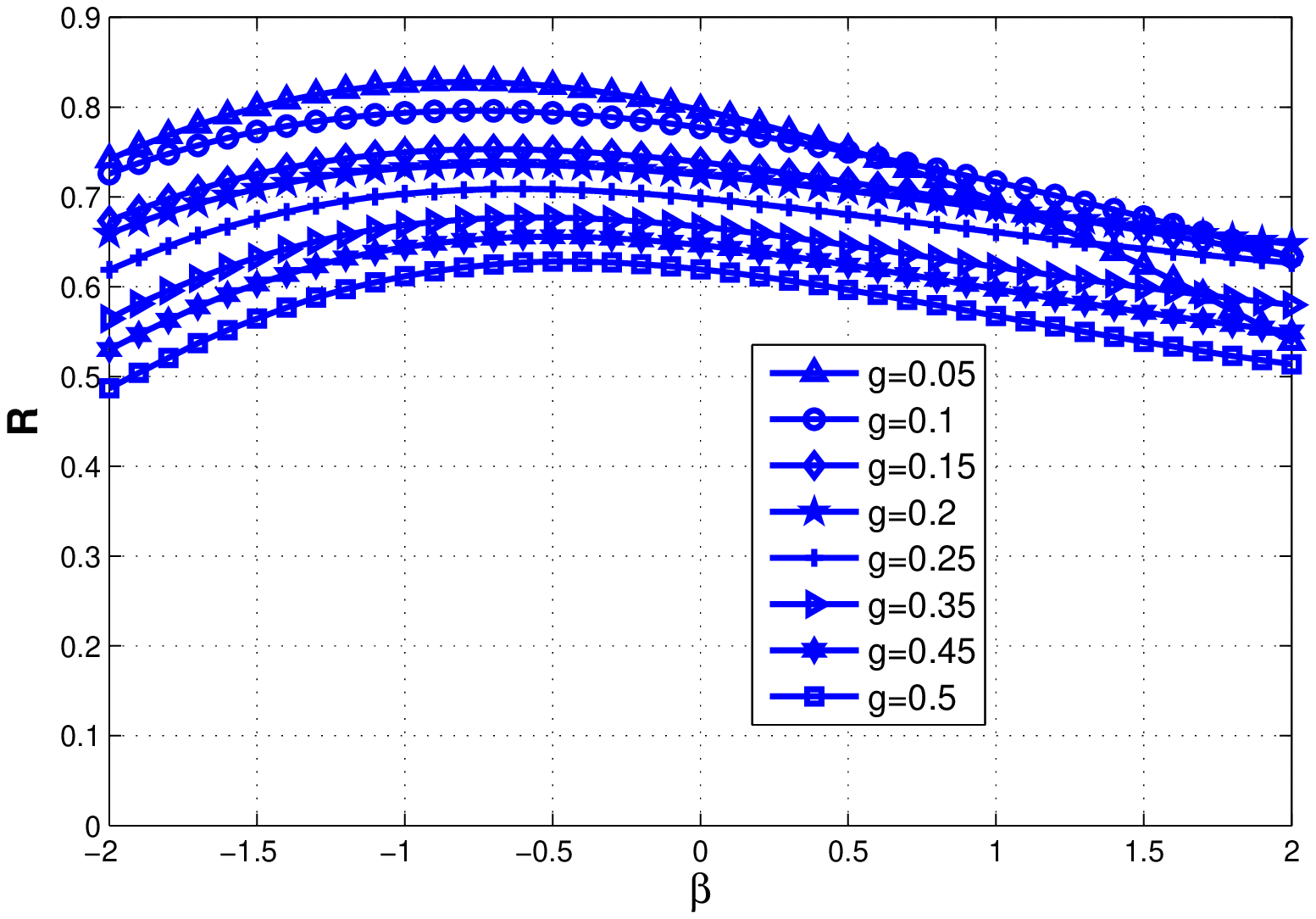} &
\includegraphics[width=1.8in, height=1.4 in]{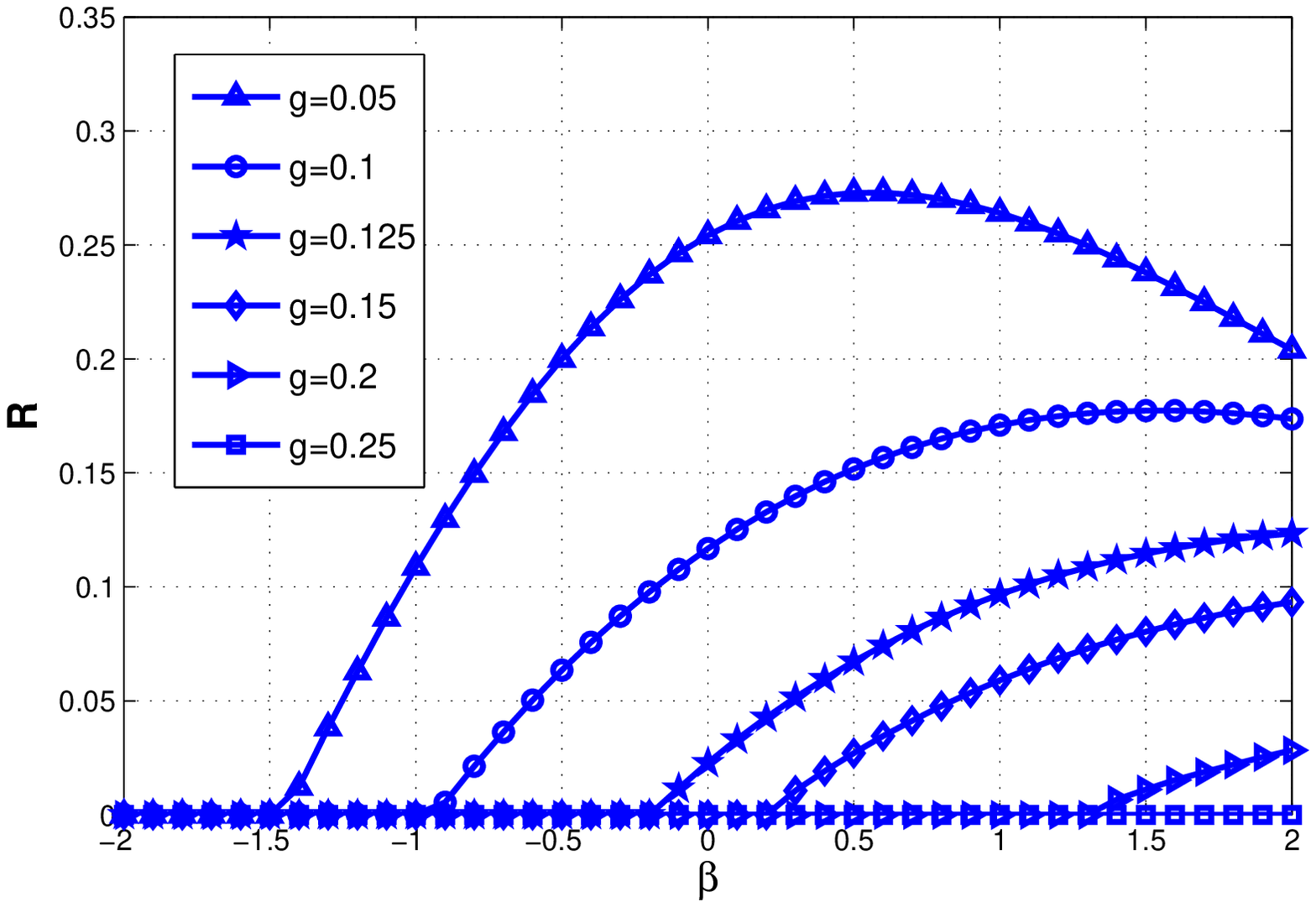} &
\includegraphics[width=1.8in, height=1.4 in]{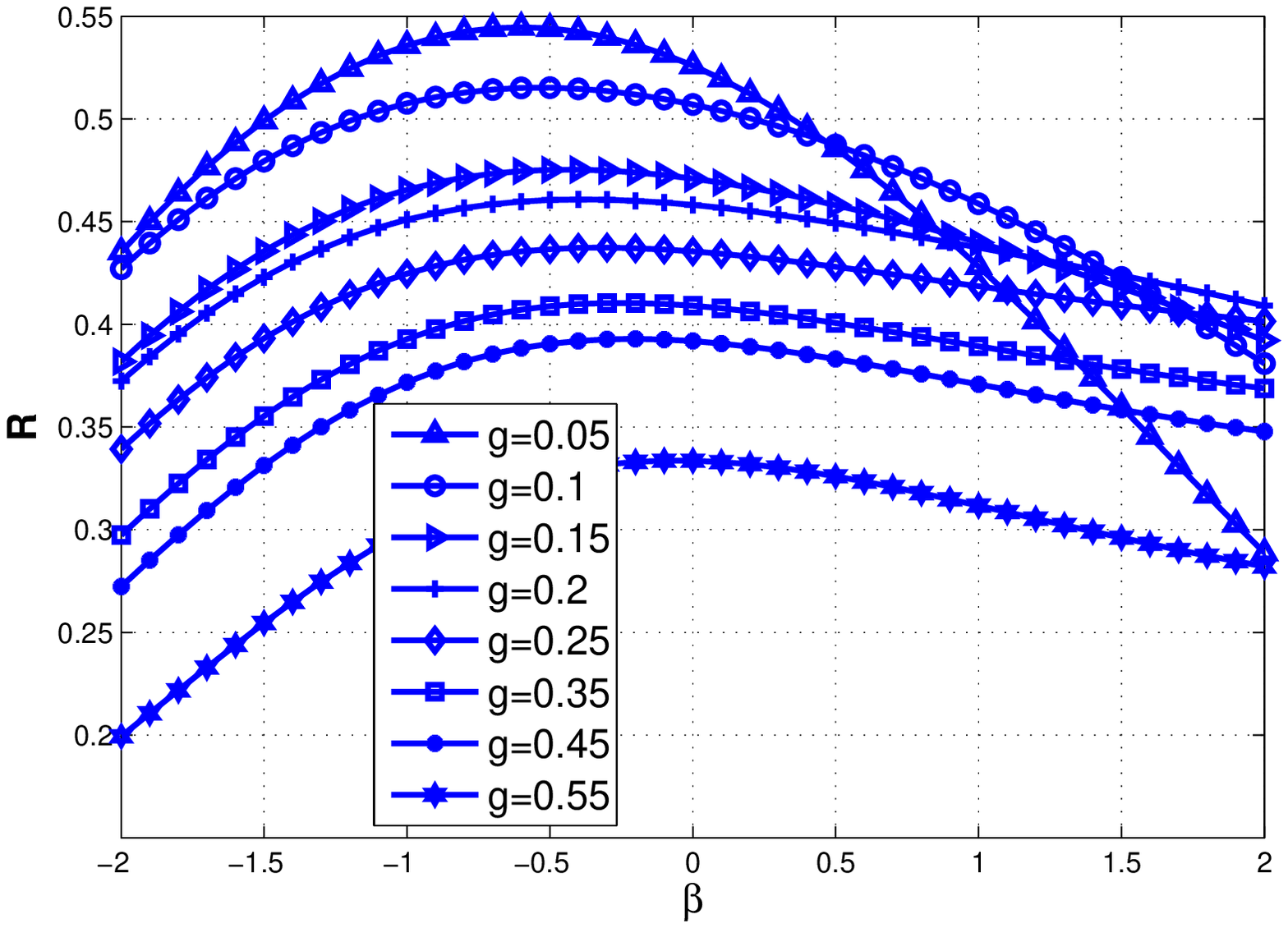}\\
\mbox{\textbf{(a)} $ \alpha=1 $ and $ \lambda=1 $} & \mbox{\textbf{(b)}  $ \alpha=1 $ and $ \lambda=0.5 $} & \mbox{\textbf{(c)}  $ \alpha=0.5 $ and $ \lambda=1 $}\\
\end{array}$
\end{center}
\caption{Final size of rumor \textit{R} vs $ \beta $ in targeted inoculation scheme for different fraction of inoculation (\textit{g})} \label{R_beta_targ}
\end{figure}

 \begin{figure}[h]
\begin{center}
\includegraphics[scale=0.32]{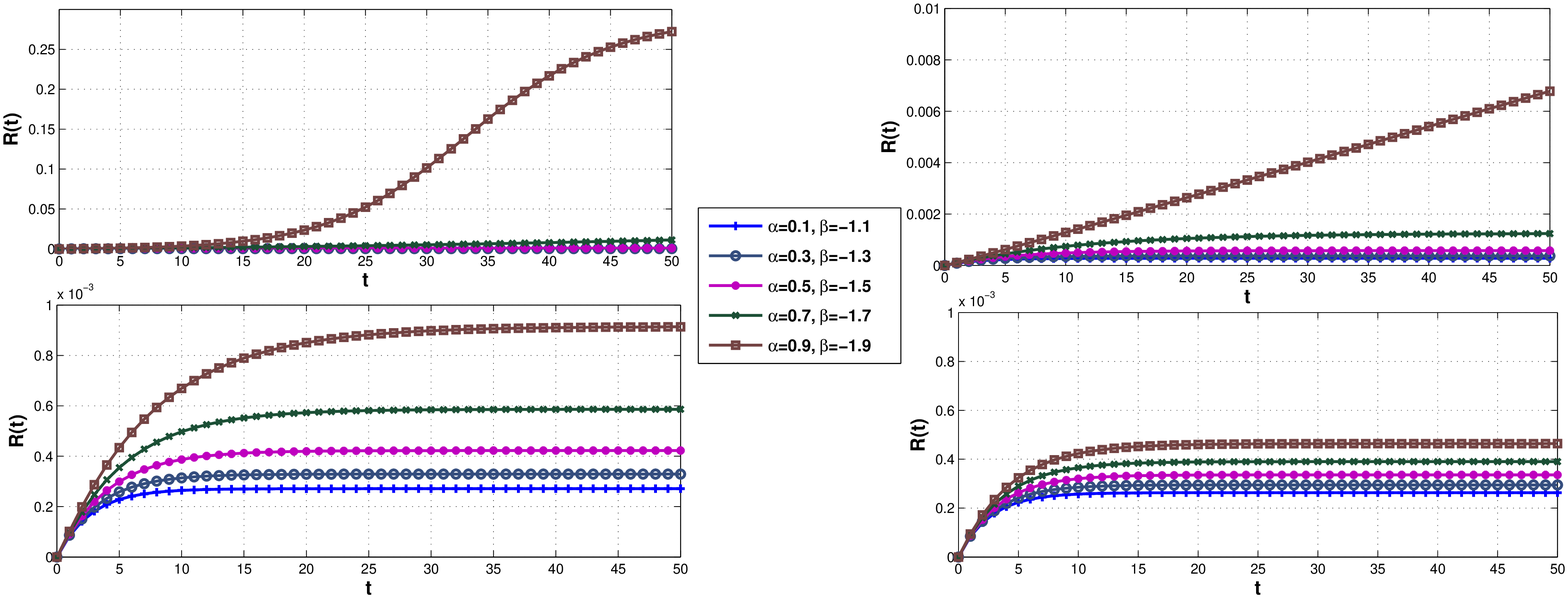}
\end{center}
\caption{ \textit{R(t)} vs \textit{t} with $ \alpha+\beta =-1 $ , $ \lambda=0.6 $ in random inoculation for \textit{g}= 0.1, 0.3 (upper) 0.5, 0.7 (lower)} \label{Rrt1}
\end{figure} 

\begin{figure}[h]
\begin{center}
\includegraphics[scale=0.32]{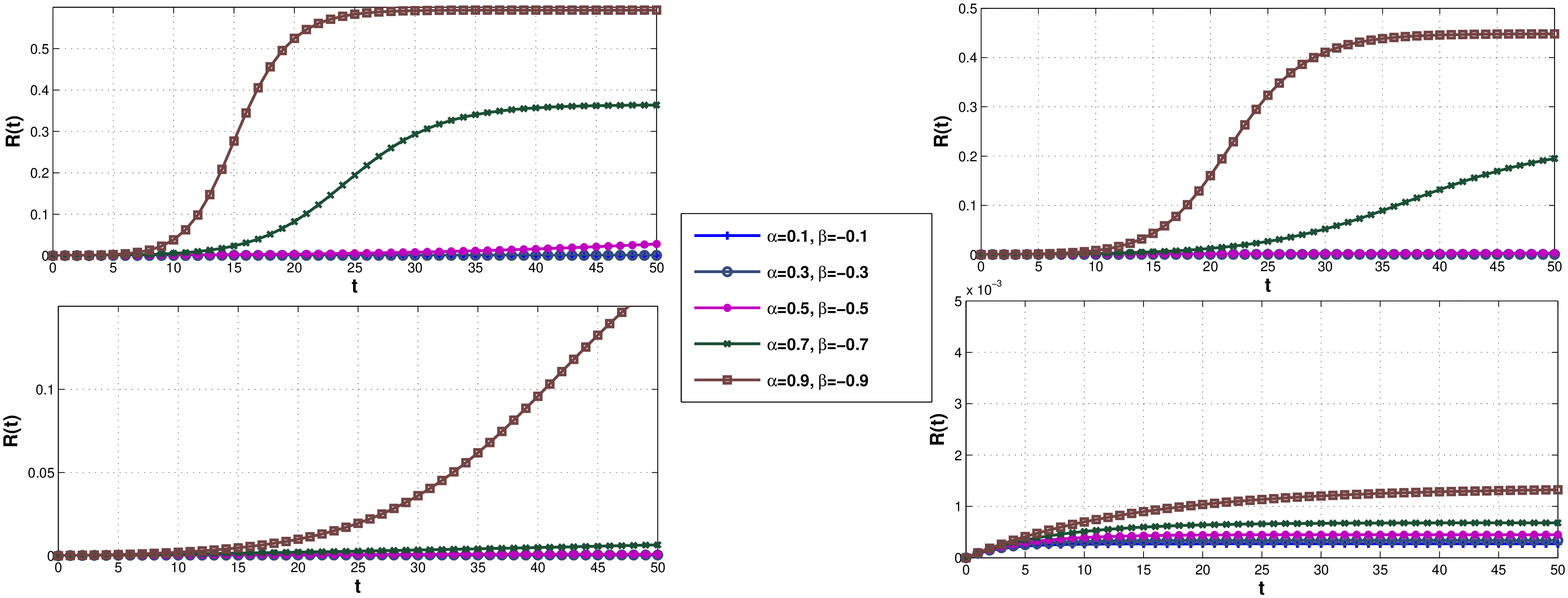}
\end{center}
\caption{ \textit{R(t)} vs \textit{t} with $ \alpha+\beta =0 $ , $ \lambda=0.6 $ in random inoculation for \textit{g} = 0.1, 0.3 (upper) 0.5, 0.7 (lower)} \label{Rrt2}
\end{figure} 

\begin{figure}[h]
\begin{center}
\includegraphics[scale=0.32]{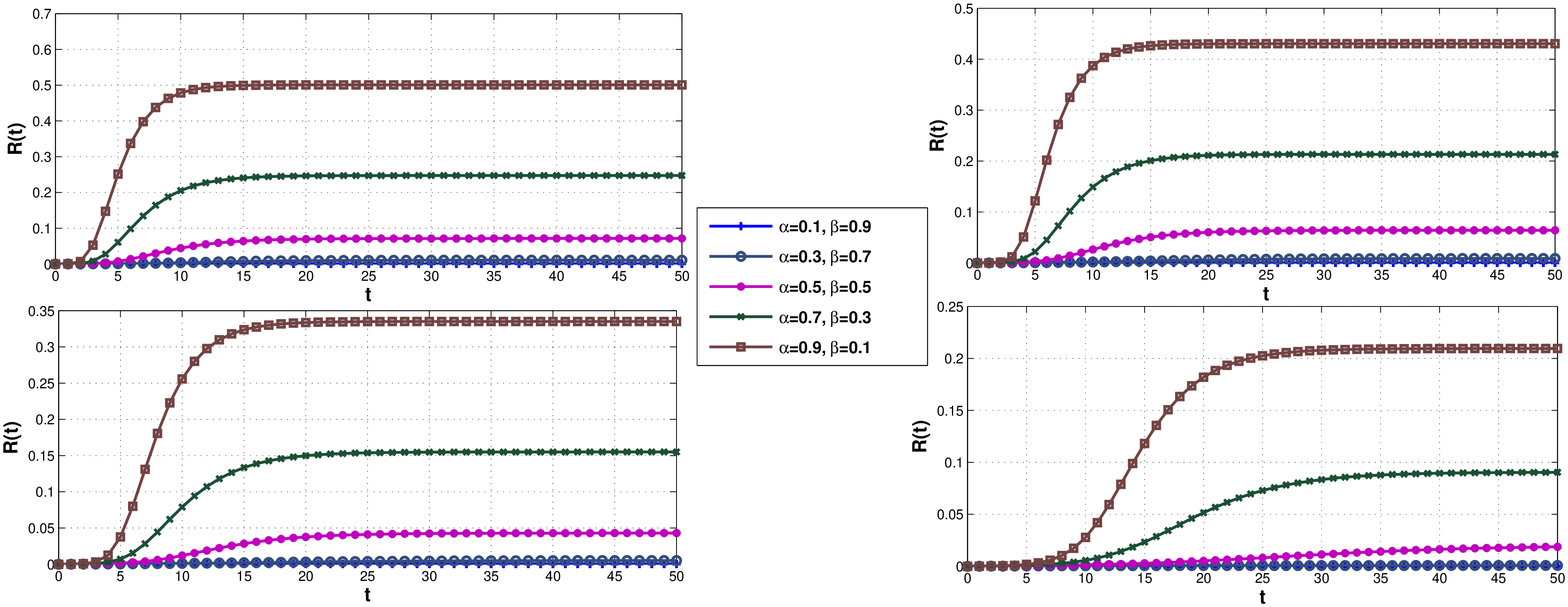}
\end{center}
\caption{\textit{R(t)} vs \textit{t} with $ \alpha+\beta =1 $ , $ \lambda=0.6 $ in random inoculation for \textit{g} = 0.1, 0.3 (upper) 0.5, 0.7 (lower) } \label{Rrt3}
\end{figure} 

\begin{figure}[h]
\begin{center}
\includegraphics[scale=0.32]{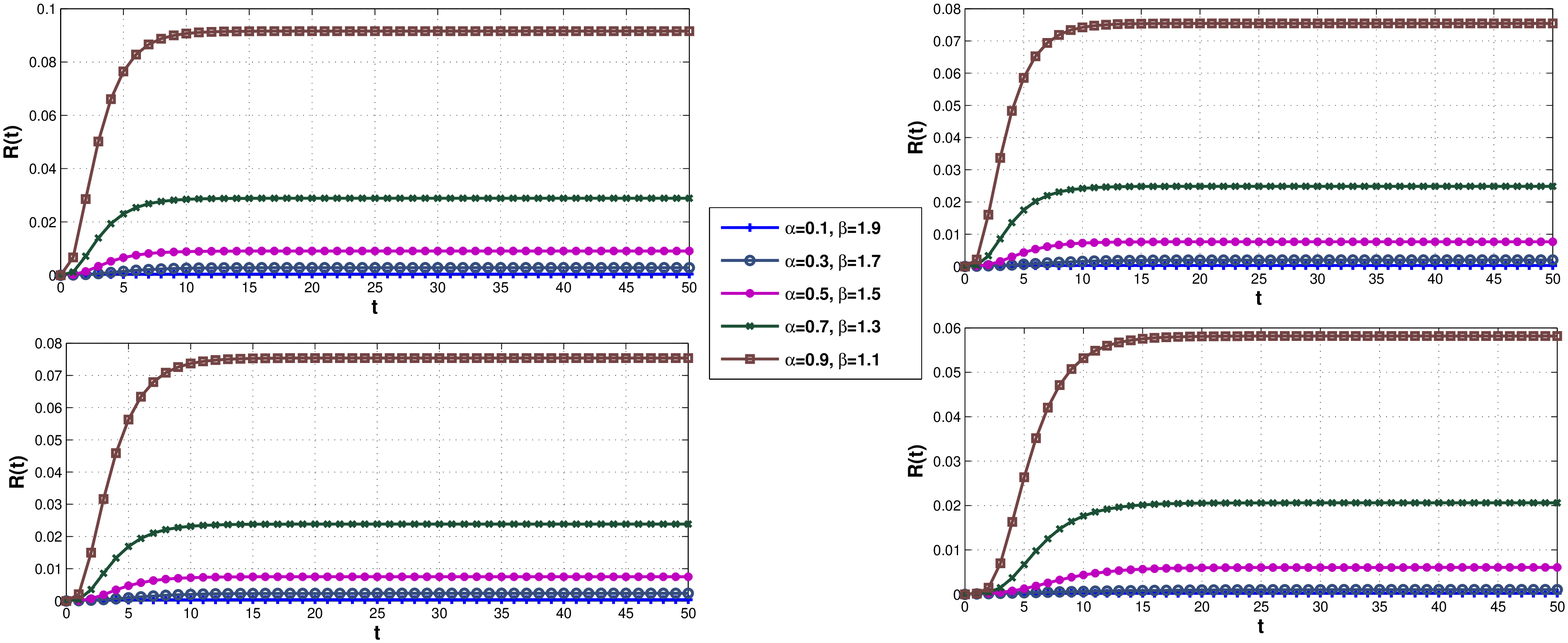}
\end{center}
\caption{ \textit{R(t)} vs \textit{t} with $ \alpha+\beta =2 $ , $ \lambda=0.6 $ in random inoculation for  \textit{g}= 0.1, 0.3 (upper) 0.5, 0.7 (lower) } \label{Rrt4}
\end{figure}

\begin{figure}[h]
\begin{center}
\includegraphics[scale=0.32]{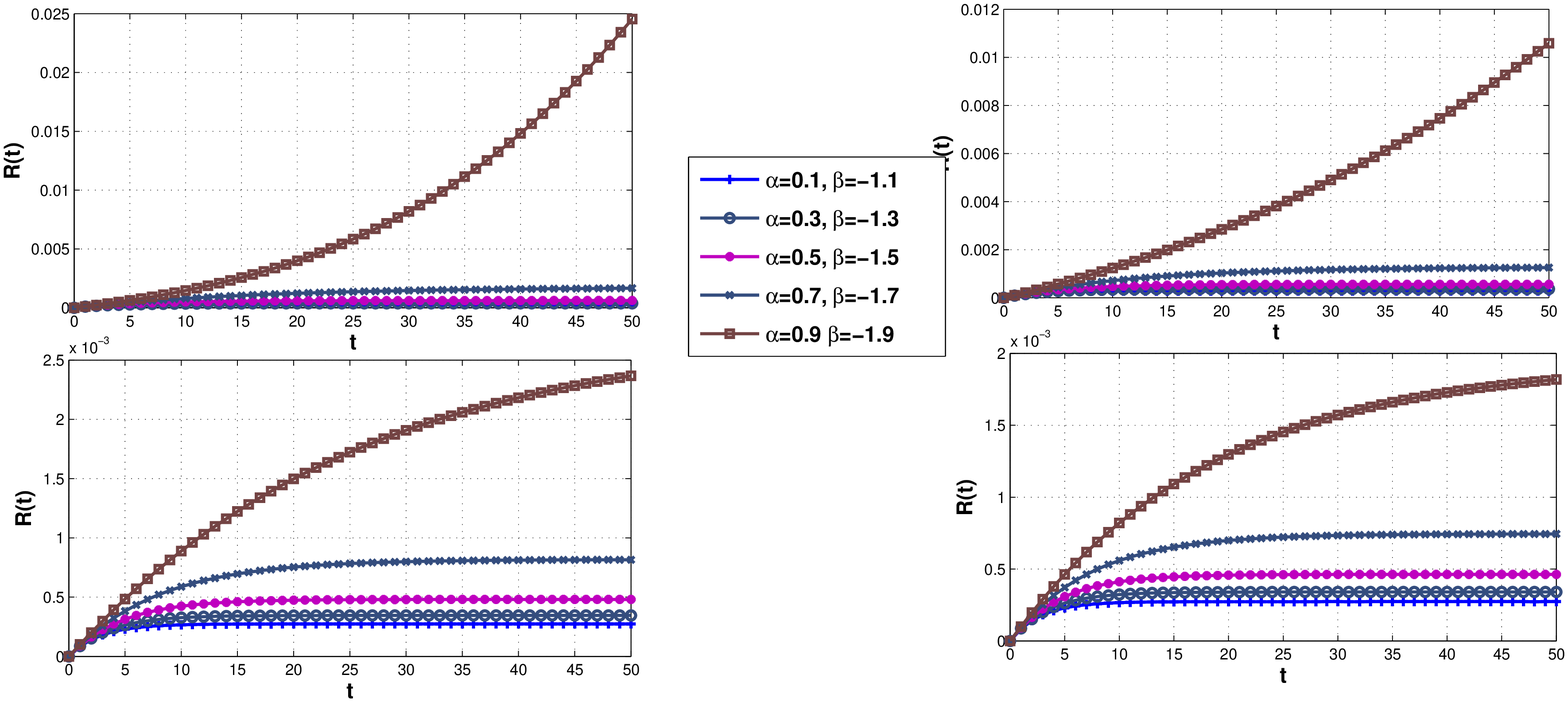}
\end{center}
\caption{\textit{R(t)} vs \textit{t} with $ \alpha+\beta =-1 $ , $ \lambda=0.6 $ in targeted inoculation for  \textit{g}= 0.05, 0.1 (upper) 0.15, 0.2 (lower)} \label{Rtt1}
\end{figure} 

\begin{figure}[h]
\begin{center}
\includegraphics[scale=0.32]{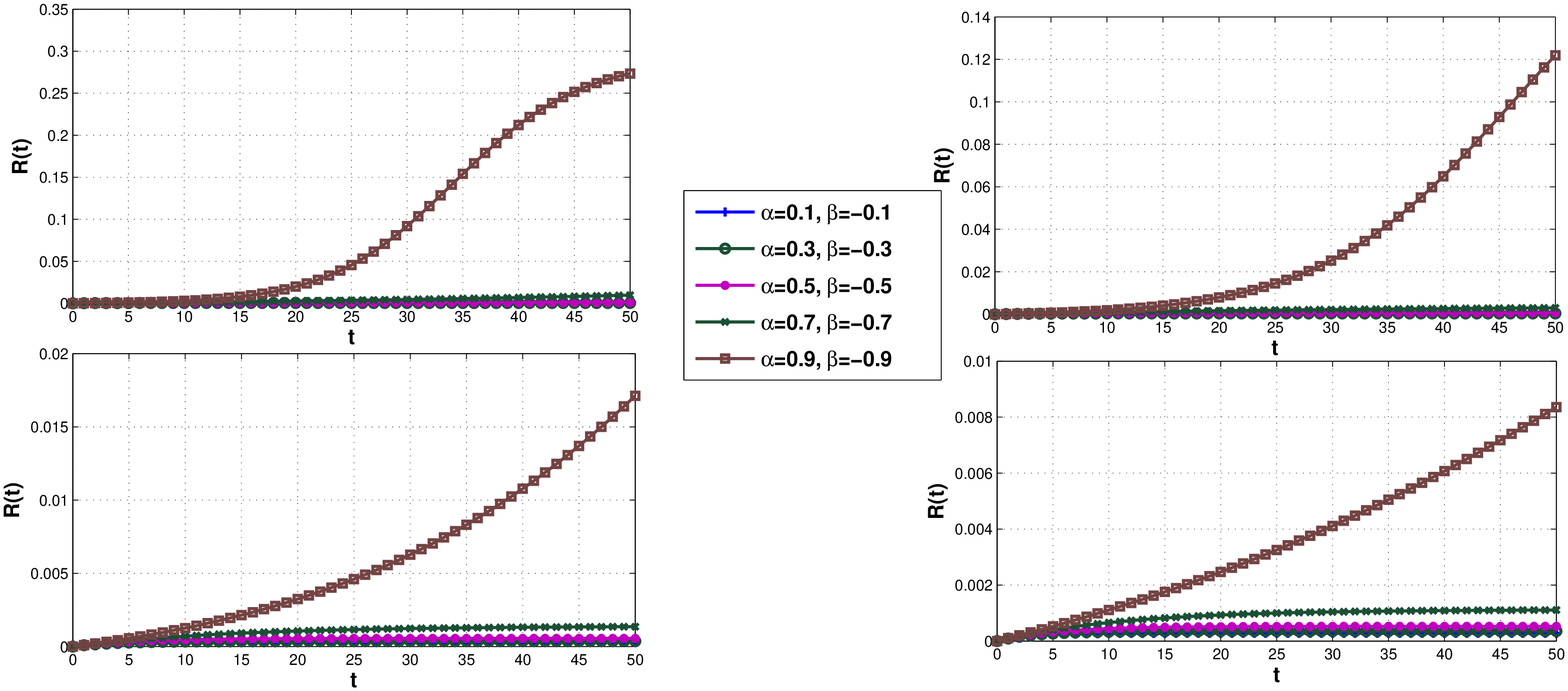}
\end{center}
\caption{\textit{R(t)} vs \textit{t} with $ \alpha+\beta = 0 $ , $ \lambda=0.6 $ in targeted inoculation for  \textit{g}= 0.05, 0.1 (upper) 0.15, 0.2 (lower) } \label{Rtt2}
\end{figure} 

\begin{figure}[h]
\begin{center}
\includegraphics[scale=0.32]{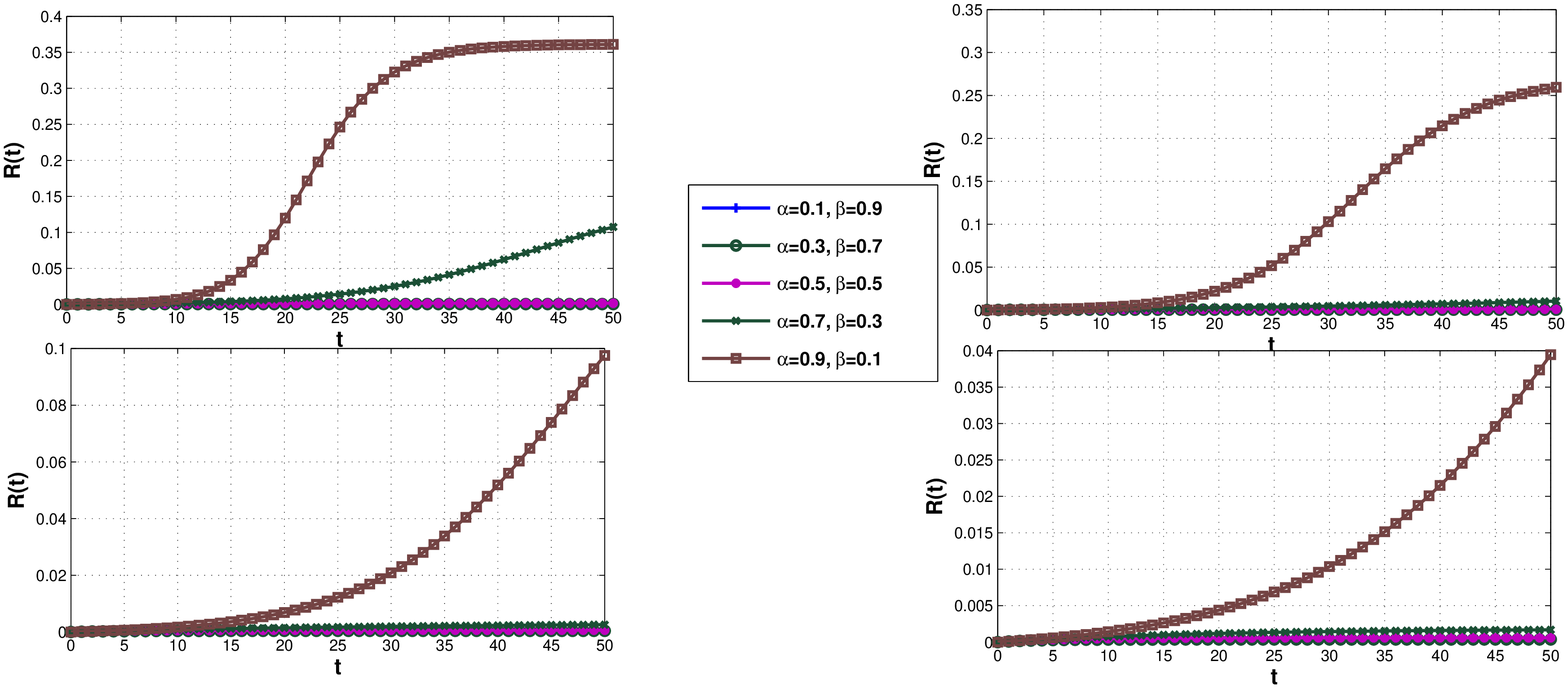}
\end{center}
\caption{ \textit{R(t)} vs \textit{t} with $ \alpha+\beta =1 $ , $ \lambda=0.6 $ in targeted inoculation for  \textit{g}= 0.05, 0.1 (upper) 0.15, 0.2 (lower) } \label{Rtt3}
\end{figure} 

\begin{figure}[h]
\begin{center}
\includegraphics[scale=0.32]{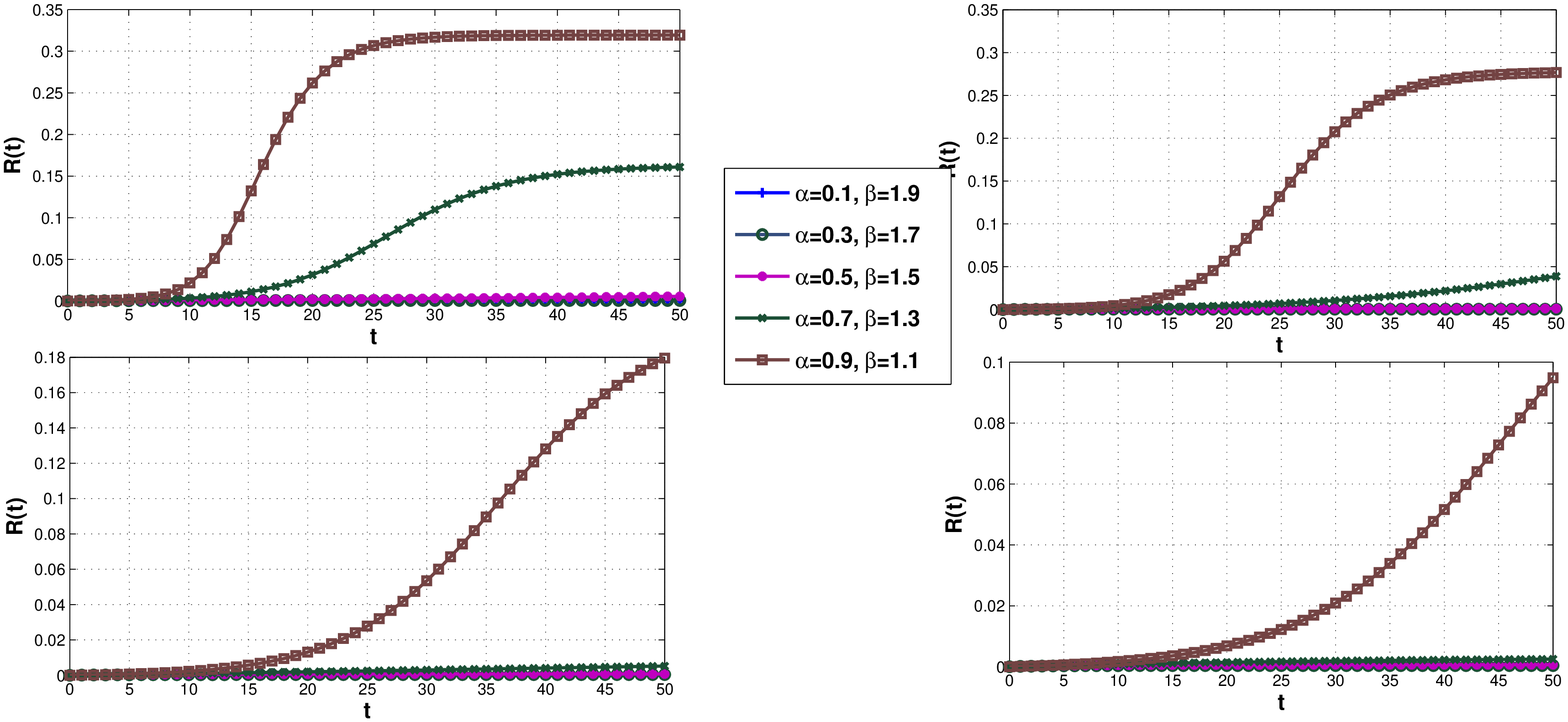}
\end{center}
\caption{ \textit{R(t)} vs \textit{t} with $ \alpha+\beta =2 $ , $ \lambda=0.6 $ in targeted inoculation for  \textit{g}= 0.05, 0.1 (upper) 0.15, 0.2 (lower) } \label{Rtt4}
\end{figure}

\section{Conclusion}
In presented study, the modified SIR model has been proposed by considering standard SIR rumor spreading model with degree dependent tie strength of nodes and nonlinear spread of rumor. The two parameters nonlinear exponent $ \alpha $ and degree dependent tie strength exponent $ \beta $ have been introduced. In modified rumor spreading model, finite rumor spreading threshold has been found for finite scale free networks while fixed rumor threshold has been found for any size of network when $ \alpha+\beta+2<\gamma $. Random and targeted inoculation schemes have been introduced in the proposed modified model. Rumor threshold in targeted inoculation scheme is found to be higher than the random inoculation. On the other hand the rumor threshold in random inoculation is higher than the modified model without inoculation. It has also been observed that for scale free networks targeted inoculation scheme is successful in suppressing the rumor spreading in the network, since it requires to inoculate less number of 
nodes than random inoculation. Further $ \alpha $ is found to be more sensitive than $ \beta $, as it affects more to rumor threshold. Finally it is seen that in real world networks finite rumor threshold can be achieved by considering more realistic parameters (degree dependent tie strength of nodes and nonlinear spread of rumor). The targeted inoculation scheme can be successfully applied to suppress the rumor spreading over scale free networks.

\bibliographystyle{acm}
\bibliography{paper_anurag}

\end{document}